\documentclass[12pt]{amsart}

\usepackage{mathptmx}
\usepackage{helvet}
\usepackage{courier}
\usepackage[margin=2.5cm]{geometry}
\usepackage{graphicx}
\usepackage{amsmath,amssymb,bm}

\usepackage{amscd}
\usepackage{verbatim}
\usepackage{enumerate}

\newtheorem{theorem}{Theorem}

\newtheorem{proposition}{Proposition}
\newtheorem{lemma}{Lemma}
\newtheorem{remark}{Remark}

\newtheorem{corollary}{Corollary}

\newcommand{\g}{\mathfrak{g}}

\newcommand{\hh}{\mathfrak{h}}
\newcommand{\kf}{\mathfrak{k}}
\newcommand{\af}{\mathfrak{a}}

\newcommand{\RR}{\mathbb{R}}
\newcommand{\CC}{\mathbb{C}}

\newcommand{\pa}{\partial}

\newcommand{\be}{\begin{equation}}
\newcommand{\ee}{\end{equation}}

\newcommand{\cM}{\mathcal{M}}
\newcommand{\cS}{\mathcal{S}}

\newcommand{\cO}{\mathcal O}
\newcommand{\cP}{\mathcal P}

\newcommand{\cB}{\mathcal B}

\DeclareMathOperator{\Ad}{Ad}

\title{Periodic and open classical spin Calogero-Moser chains}

\author{Nicolai Reshetikhin}
\address{N.R.: YMSC, Tsinghua University, Beijing; Department of Mathematics, University of California, Berkeley, \& Physics Department, St. Petersburg University.}
\email{reshetik@math.berkeley.edu}

\begin{document}

\begin{abstract}
We construct a class of interacting spin Calogero-Moser type 
systems. They can be regarded as a many particle system with spin degrees of freedom
and as an integrable spin chain of Gaudin type.
We prove that these Hamiltonian systems are superintegrable.

\end{abstract}

\maketitle
\section*{Introduction} 

\subsection*{1}Classical Calogero-Moser (CM) systems were among the first integrable
$N$-particle systems of one dimensional particles \cite{Ca}\cite{Mo} with the potential
$1/(q_i-q_j)^2$. This model was generalized to the potential $1/sh^2(q_i-q_j)$ in \cite{Suth}.
Then it was extended to other root systems and to elliptic potentials in \cite{OP}, to a model involving spin degrees of freedom in \cite{GH}.

There is an extensive literature on spin versions of CM systems. For example in \cite{KBBT}\cite{KrZ} solutions to equations of motion to the elliptic spin Calogero-Moser system were related to special elliptic solutions to the matrix KP hierarchy. The relation to gauge theories were explored
in many papers, see for example \cite{Nekr}\cite{FoR}. A variety of spin CM systems were obtained by L. Feher, see for example \cite{FP1}\cite{FP2}\cite{Fe1}, in particular he
derived important examples related to homogeneous spaces. Two spin
CM systems were studied in \cite{KLOZ1}\cite{KLOZ2}. Integrable chains of relativistic 
spin CM type systems were studied in \cite{ChF}\cite{AO}.

Superintegrability of spin CM systems and of spin Ruijsenaars systems was established
in \cite{R1}. In \cite{R3} the superintegrability of spin CM systems on homogeneous spaces was established.  A family of superintegrable systems
on moduli spaces of flat connections was constructed in \cite{AR}. This family includes systems
studied in \cite{ChF}\cite{AO}. In these particular case the system is also Liouville integrable. 

In this paper  we will describe classical superintegrable system which we call spin Calogero-Moser(CM) chains.
We call them spin CM chains because they combine features of many particle systems (as in CM systems) and of spin chains. We distinguish two cases: a {\it periodic chain} and  an {\it open chain}. 
The periodic case is the classical version of a quantum integrable system where 
joint eigenfunctions of quantum commuting Hamiltonians are trace functions, see \cite{ES0}. 
In this case the spin part of the system reminds a spin chain with periodic boundary conditions. 
In case of rank 1 orbits for $\mathfrak{sl}_n$ these systems are linearized versions of \cite{ChF} and \cite{AO}. In the open case they are a classical version of quantum integrable systems constructed in \cite{SR}\cite{RS}. For these systems the spin part of the system is similar to an open spin chain.

In both cases, i.e. in the periodic and in the open spin Calogero-Moser chains, the phases space
is a stratified symplectic space \cite{LS}, which, in some cases have only one stratum and becomes a
symplectic manifold. 

\subsection*{2} Recall that a {\it superintegrable system} is the structure 
on a symplectic manifold $\cM$ that consists of a Poisson manifold $\cP$, a Poisson manifold $\cB$ with the trivial Poisson structure (i.e. zero Poisson tensor) and two surjective Poisson projections
\begin{equation}\label{DIT}
\cM\stackrel{p_1}{\rightarrow} \cP \stackrel{p_2}{\rightarrow} \cB
\end{equation}
such that $\dim(\cM)=\dim(\cP)+\dim(\cB)$.
For a superintegrable system a generic fiber of  $p_1$ is an isotropic submanifolds of dimension $\dim(\cB)$ and connected components of a generic fiber of $p_2$ is a disjoint union of symplectic leaves of $\cP$. For details see \cite{N}, \cite{R2} and references therein. Here we adopt this notion to the case of stratified symplectic and Poisson spaces in which case $p_1$ and $p_2$ are Poisson mapping between stratifies spaces. In this paper the superintegrability means the balance of dimensions for the big stratum.
How the system behave at smaller strata will be a subject of a separate publication. In the algebraic case,
the appropriate setting is symplectic and Poisson stacks. 

Let $I$ be a Poisson commutative subalgebra of $A=C^\infty(\cM)$ that consists of 
functions which are constant on fibers of $p_2\circ p_1$ (the pull-back of functions on $\cB$ to functions on $\cM$) and $J$ be the Poisson algebra of functions which are constant on fibers of $p_1$ (the pull-back of functions on $\cP$). The condition on $(\cM, \cP, \cB)$ for being a superintegrable system  is equivalent to the following condition on $I\subset J\subset A$. The Poisson algebra $A$ has trivial center, $I\subset A$ is
a Poisson commutative subalgebra, such that  $J$, its centralizer in $A$ maximal possible Gelfand-Kirillov dimension for the given Gelfand-Kirillov dimension of $I$.
%rank for the given rank
%of $I$. Here rank of the commutative algebra $I$ is the maximal 
%number of independent elements in it. 

The Hamiltonian dynamics generated by a function $H\in I$ is called {\it superintegrable}. 
Any function from $J$ is constant along flow lines of the vector field generated by $H$ and thus, it is an integral of motion for the Hamiltonian dynamics generated by $H$. 
This is why we call elements of the Poisson commutative subalgebra $I$ Hamiltonians and elements of $J$ conservation laws.

\subsection*{3}Throughout this paper $G$ is a split real connected semisimple Lie algebra with finite center which admits a complexification, and  $\Theta\in\textup{Aut}(G)$ is a Cartan involution.
We denote by $K=G^{\Theta}$ the closed subgroup of fixed points of $\Theta$, which is connected and maximal compact.
%$\mathfrak{g}$ is a split real semisimple Lie algebra and $G$ is a connected real Lie group 
%Let $\Theta$ be a global Cartan involution of $G$. 
%$G$ is a split real connected semisimple Lie group with finite center
%and $\Theta$ is a global Cartan involution. 
%We denote by $K=G^\Theta$ the closed subgroup of fixed points of $\Theta$, which is connected and maximal compact. 
Let $\theta$ the corresponding Cartan involution\footnote{Recall that an involution $\theta:\g\to\g$ is 
a Cartan involution when the bilinear from $(-\theta(x),y)$ on $\g$ is positive definite.
Here $(\cdot,\cdot)$ is the Killing form.} of $\g$, and $\kf$ the Lie algebra of $K$. The associated Cartan decomposition of
$\mathfrak{g}$ is $\mathfrak{g}=\mathfrak{k}\oplus\mathfrak{p}$, with $\mathfrak{p}$ the $(-1)$-eigenspace of $\theta$.
%%It lifts uniquely to the global Cartan involution $\Theta: G\to G$.
%%Finiteness of the center guarantees that $K=G^\Theta$ is maximal compact in $G$.}.
%%Clearly $\kf=\g^\theta$. 

Let $\af\subset \mathfrak{g}$ be maximally noncompact $\theta$-stable Cartan subalgebra of $\g$. Since $\g$ is split we have $\mathfrak{a}\subseteq\mathfrak{p}$.
%%$\theta(a)=-a$ for all $a\in\af$.
%%Note that complex linear extension  of $\theta$ to $\g_\CC$ is the Chevalley involution 
On the Lie group level, $A=\exp(\af)\subset G$ is a maximal real split torus in $G$ and $H:=Z_G(A)$, the centraliser of $A$ in $G$, is a Cartan subgroup in $G$ containing $A$.
The exponential map provides an isomorphism $\af\overset{\sim}{\longrightarrow} A$, 
whose inverse we denote by $\log: A\rightarrow\af$.

Consider the root decomposition of $\g$ with respect to the Cartan subalgebra $\af$,
\[
\g=\af\oplus \bigoplus_{\alpha\in R} \g_\alpha
\] 
where $R\subset\mathfrak{a}^*$ is the root system of $\mathfrak{g}$ relative to $\mathfrak{a}$. Choose $e_\alpha\in \g_\alpha$ such that 
\begin{equation}\label{RB}
\theta(e_\alpha)=-e_{-\alpha}
%, \ \ \theta(a)=-a, \ \ a\in \af
\end{equation}
and $(e_\alpha, e_{-\alpha})=1$ for each $\alpha\in R$, and choose a subset $R_+\subset R$ of positive roots. Let
$W\subset\textup{GL}(\mathfrak{a}^*)$ be the Weyl group of $R$.

The Weyl group $W$ is isomorphic to $N_G(A)/H$, where
$N_G(A)$ is the normaliser of  $A$ in $G$.
Denote by $A_{reg}$ the set of regular elements in $a\in A$,
\[
A_{reg}:=\{a\in A \,\, | \,\, a_\alpha:=e^{\alpha(\log(a))}\not=1 \hbox{ for all } \alpha\in R\}.
\]
It is the union of all the regular $W$-orbits in $A$. A fundamental domain for the $W$-action on $A_{reg}$
\footnote{In case of $SL_n(\RR)$ one can take $A$ to be the diagonal unimodular matrices with positive real entries, and $A_{reg}$ consists of those with distinct diagonal entries.}
 is 
%$A_+\subset A$ the fundamental domain of this action.
%In other words 
the positive Weyl chamber 
\[
A_+=\{a\in A \,\, | \,\, a_\alpha:=e^{\alpha(\log(a))} > 1 \mbox{ for any } \alpha\in R_+\}.
\]
Let $G^\prime\subset G$ be the set of elements
$g\in G$ which are $G$-conjugate to some element in $A_{reg}$.
%Note that $G^\prime$ consists of regular elements in $G$. 
The inclusion
$A_{reg}\hookrightarrow G^\prime$ induces a bijection $A_{reg}/W\overset{\sim}{\longrightarrow}G^\prime/G$, with $A_{reg}/W$ the set of $W$-orbits in $A_{reg}$ and
$G^\prime/G$ the set of $G$-conjugacy classes in $G^\prime$.
%where $a_\alpha:=e^{\alpha(\log(a))}$. 
%The closure of $A_+$ is a fundamental domain for the $W$-action on $A$. 
%We set $A_{reg}$ of regular $W$-orbits in $A$ is $WA_+$.

The Weyl group $W$ is also isomorphic to $N_K(A)/M$, where $N_K(A)=N_G(A)\cap K$ and $M=H\cap K$
%is the normaliser of $A$ in $K$ and $M=Z_G(A)\cap K$ the
%centraliser of $A$ in $K$. 
(note that $M$ is a finite group since $G$ is split). The inclusion map $A\hookrightarrow G$ induces an isomorphism $A/W\overset{\sim}{\longrightarrow} K\backslash G/K$.
We write $G_{reg}=KA_+K$ for the union of the double $(K,K)$-cosets intersecting $A_{reg}$.
%In particular, $G=KAK$, which is the Cartan decomposition.
%admits global realisation as $N_{G_{\mathbb{C}}}(H)$, 
%Let $R_+\subset R$ be positive roots and $N\subset G$ be the unipotent subgroup 
%corresponding to positive roots. The Iwasawa decomposition is the global decomposition 
%$G=KAN$.
%Let $W=N_K(A)/Z_K(A)$ be the (analytic) Weyl group of $G$. Here $N_K(A)=N_G(A)\cap K$ is the 
%normalizer of $A$ in $K$ and $Z_K(A)=Z_G(A)\cap K$ is the 
%centralizer of $A$ in $K$. The Weyl group acts naturally on $A$ and on $\af$.
% by reflections with respect to roots.
%Define $A_{reg}$ as $WA_+$, which is the set of regular $W$-orbits in $A$. 
%The group $G$ has (global) Cartan decomposition $G=KAK$. 
%We have a natural isomorphism $K\backslash G/K\simeq A/W$. This refines the Cartan decomposition.
%which is also known as a refined Cartan decomposition.

\subsection*{4} The phase space of a periodic spin Calogero-Moser chain corresponding to a collection $\cO=\{\cO_1,\dots, \cO_n\}$ of coadjoint orbits $\cO_i\subset \g^*$ is the regular part of the symplectic leaf $\cS(\cO)$ of the stratified Poisson space $T^*(G^{\times n})/G_n$, with the action of the gauge group $G_n=G^{\times n}$ the lift of a twisted conjugation action on $G^{\times n}$(see section \ref{g-act-per}).\footnote{In this paper we assume that all quotients $X/H$ are GIT quotients.} Here we assume that each of $\cO_i$ is non-trivial, i.e. $\cO_i\neq \{0\}$
%$T^*(G^{\times n})/G^{\times n}$
 %Here the action of  $G^{\times n}$ is the lift of the gauge action of this group on $G^{\times n}$(see section \ref{g-act-per}). 
 These symplectic leaves are obtained by the Hamiltonian reduction, as it is described in section \ref{g-act-per}. As a stratified symplectic space
 \[
 \cS(\cO)\simeq \{(x_1,\dots, x_n, g)\in {\g^*}^{\times n}\times G\,\, | \,\, x_1-\textup{Ad}_{g^{-1}}^*x_n\in\mathcal{O}_1,\,\, x_{i}-x_{i-1}\in\mathcal{O}_i,\,\, i=2,\ldots,n \}/G,
 %\cS(\cO)\simeq \{(x_1,\dots, x_n, g)\in {\g^*}^{\times n}\times G\,\, | \,\, x_{i+1}-x_i\in\mathcal{O}_i,\,\, i=1,\ldots,n-1,\,\, x_1-\textup{Ad}_{g^{-1}}^*x_n\in\mathcal{O}_n\}/G
% \cS(\cO)\simeq \{(x_1,\dots, x_n, g)\in {\g^*}^{\times n}\times G|
 % x_i-Ad^*_g(x_{i+1})\in \cO_i\}/G\subset ({\g^*}^{\times n}\times G)/G.
 \]
 %where $i$ is taken modulo $n$, 
see section \ref{ph-space-reg}.\footnote{There are $n$ such natural isomorphisms $\varphi_j$ ($1\leq j\leq n$), see section \ref{g-red}. In the introduction we use $\varphi_n$.} 
Its regular part is defined as the intersection $\cS(\cO)_{reg}=\cS(\cO)\cap ({\g^*}^{\times n}\times G')/G$.\footnote{A better way to think about the periodic spin Calogero-Moser system for a real split simple Lie group $G$ is to define $\cS(\cO)_\CC$ is complex algebraic setting and then to take the corresponding real slice. This will be addressed in another publication.}
 %and the regular part is its "biggest" stratum.}. 
% were that $G'$ is the union of conjugation orbits through $A_{reg}$.  
The regular part has the following structure as a symplectic
manifold, $\cS(\cO)_{reg}\simeq \bigl(\nu_{\cO}^{-1}(0)/H\times T^*A_{reg}\bigr)/W$, where $\nu_\cO:\cO_1\times\cdots\cO_n\rightarrow\af^*$ is the moment map
for the diagonal coadjoint action of $H$ on $\cO_1\times\cdots\times\cO_n$ (see section \ref{ph-space-reg}).
%$S(\cO)_{reg}\simeq (T^*A_{reg}\times (\cO_1\times\dots\times \cO_n)\sslash A)/W$ 
%where 
%$\cO_i$ are coadjoint orbits in $\g^*$, 
%$X\sslash A$ is the Hamiltonian reduction with respect to a Hamiltonian action of $A$ on $X$. 
% and $W$ is the Weyl group of $G$. 

%Here and in the text $\cO=\{\cO_1, \dots, \cO_n\}$ is a collection
%of coadjoint orbits in $\g^*$ and $S(\cO)_{reg}\subset S(\cO)$ is the regular part of the symplectic leaf $S(\cO)\subset T^*(G^{\times n})/G^{\times n}$ where the quotient is taken with respect to the action of the gauge action of $G^{\times n}$ on itself (see section \ref{pCM}).
Trivialization  of $T^*G$ by right translations gives an isomorphism $T^*(G^{\times n})\simeq {\g^*}^{\times n} \times G^{\times n}$ and the Poisson projection $T^*(G^{\times n})/G_n\to {(\g^*/G)}^{\times n}$, which is the projection to the cotangent directions followed by the quotienting with respect to the coadjoint action of $G^{\times n}$. Poisson commuting Hamiltonians of the periodic spin Calogero-Moser system are functions on $T^*(G^{\times n})/G_n$
which are constant on fibers of this Poisson projection. More precisely, tPoisson commuting Hamiltonians are such functions restricted to $\cS(\cO)_{reg}$. 

Consider the increasing set of natural numbers $2=d_1\leq \cdots\leq d_r$ with $r=\textup{rank}(\g)$ and $d_k-1$ the exponents of $\g$. Let $c_{d_k}$ be the nonzero coadjoint invariant functions on $\g^*$ of degree $d_k$,
%Let $d_k, \ \ k=1,\dots, r=renk(\g)$ be exponents of $\g$ and $c_{d_k}$ be coajoint invariant functions on $\g^*$ of degree $d_k$. 
known as Casimir functions. The function $c_2$ is the quadratic Casimir of $\g$. 
Let  $H^{(l)}_{d_k}$ be the function on $({\g^*/G})^{\times n}$ which is $c_{d_k}$ on the $l$-th factor and
constant on all other factors.
%factors except the $l$-th factor, and is  

Let us denote vectors in $\cO_k\subset \g^*$ by $\mu^{(k)}$, its Cartan component by $\mu_0^{(k)}$ and set
% its coordinate along the root vector
%$e_\alpha$ by 
$\mu_\alpha^{(k)}=\mu^{(k)}(e_{-\alpha})$ for $\alpha\in R$. Denote 
$(p,a)$ points on $T^*A\simeq \af^*\times A$. Now let us describe quadratic Hamiltonians
in terms of these variables.

The $n$-th quadratic Hamiltonian is the spin Calogero-Moser Hamiltonian. It has particularly simple form:
\[
H^{(n)}_2=\frac{1}{2} (p,p)-\sum_{\alpha>0}\frac{\mu_\alpha \mu_{-\alpha}}{2\textup{sh}^2(q_\alpha)}
\]
where 
we used the parametrization $a_\alpha=e^{q_\alpha}$ (so $q_\alpha=\alpha(\log(a))$) 
and $\mu_\alpha=\mu_\alpha^{(1)}+\cdots+\mu_\alpha^{(n)}$, and $(\cdot,\cdot)$ is the Euclidean form on $\af^*$ obtained by dualising the restriction of the Killing form of $\g$ to $\af$.

The differences $D_k=H^{(k)}_2-H^{(k-1)}_2$ for $1<k\leq n$ are classical analogs of topological
Knizhnik-Za\-mo\-lod\-chi\-kov-Bernard differential operators,
\begin{equation}\label{Dk-per-int}
D_k=(\mu^{(k)}_0,p)-\sum_{l=1}^{k-1} r_{lk}+\sum_{l=k+1}^nr_{kl}
\end{equation}
where $r_{kl}$ for $k\not=l$ is a classical version of the Felder's dynamical $r$-matrix \cite{F},
\begin{equation}\label{F-r-int}
r_{kl}=-\frac{1}{2}(\mu^{(k)}_0, \mu^{(l)}_0)+\sum_{\alpha}\frac{ \mu^{(k)}_{-\alpha} \mu^{(l)}_{\alpha}}{a_\alpha-1}
\end{equation}
and $\sum_\alpha$ stands for the sum over all the roots $\alpha\in R$. This explicit form of $D_k$ is derived in section \ref{ph-space-reg}.

The superintegrability of this system is described in section \ref{pCM-sup}.
The projection method for constructing solutions of equations of motion and angle variables are described in section \ref{solutions}.

One can choose $G$ to be a the maximal compact real form of the complexification $G_\CC$. In this case the integrable system is similar, but hyperbolic functions gets replaces by the trigonometric ones. The structure of the phase space is again a stratified symplectic space. The superintegrability of the quantum counterpart of such compact case is proven in \cite{R4}.

\subsection*{5} The phase space of an open Calogero-Moser spin chain is
the regular part of a symplectic leaf of the Poisson manifold $T^*(G^{\times n+1})/(K\times G^{\times n}\times K)$ where the action of the gauge group $K\times G^{\times n}\times K$ is described
in section \ref{op-g-reduct}, and $K\subset G$ is as above. Such symplectic leaves are given by the Hamiltonian reduction. They are parametrized by collections of coadjoint orbits $\cO=\{\cO^K_\ell, \cO_1, \dots, \cO_{n}, \cO^K_r\}$
where $\cO_i\subset \g^*$ and $\cO^K_{\ell,r}\subset \kf^*\subset \g^*$ are coadjoint orbits. We assume that none of $\cO_i$ is trivial, i.e. $\cO_i\neq \{0\}$. 

%Here $\kf$ is the Lie algebra of $K=G^\Theta$. 
We denote the corresponding symplectic leaf by $\cS(\cO)$. It is a stratified symplectic space. Using Cartan decomposition  $G=KAK$  and a "gauge fixing fixing", we define the regular part $\cS(\cO)_{reg}$ of $\cS(\cO)$ as the strtatum 
\[
S(\cO)_{reg}\simeq (T^*A_{reg}\times \cO^K_1\times \cO_1\times\dots\times \cO_{n}\times \cO^K_2))/N_K(A),
\] 
where on the right we have a natural product symplectic structure.

Similarly to the periodic case, quadratic Hamiltonians can be computed explicitly in terms of Cartan components $\mu^{(k)}_0$ and root coordinates $\mu^{(k)}_\alpha$ of vectors $\mu^{(k)}\in \cO_k$,  
coordinates $\mu'_{[\alpha]}, \mu''_{[\alpha]}$ on $\cO_\ell^K$ and $\cO_r^K$ respectively (in the basis elements $e_{[\alpha]}=e_{-\alpha}-e_\alpha\in \kf\subset \g$ for $\alpha\in R_+$), and $(p,a)\in T^*A_{reg}$.

Assuming the gauge fixing $\phi_n$ (see section \ref{op-g-reduct}) we have
\[
H^{(n)}_2=\frac{1}{2}(p,p)+\sum_{\alpha>0}\frac{(a_\alpha\mu_{[\alpha]}^\prime+\mu_{[\alpha]}^{\prime\prime}+a_\alpha(\mu_\alpha-\mu_{-\alpha}))(a_\alpha^{-1}\mu_{[\alpha]}^\prime+\mu_{[\alpha]}^{\prime\prime}+a_\alpha^{-1}(\mu_\alpha-\mu_{-\alpha}))}{(a_\alpha-a_{-\alpha})^2}
\]
%\[
%H^{(n)}_2=\frac{1}{2}(p,p)+\sum_{\alpha>0}\frac{(\mu_\alpha''-a_{-\alpha}\mu_\alpha'-a_{-\alpha}(\mu_\alpha-\mu_{-\alpha}))(\mu_\alpha''-a_{\alpha}\mu_\alpha'-a_{\alpha}(\mu_\alpha-\mu_{-\alpha}))}{(a_\alpha-a_{-\alpha})^2}.
%\]
For other quadratic Hamiltonians the differences
\[
D_k=H_2^{(k)}-H_2^{(k-1)}\qquad\quad (1\leq k\leq n)
\]
are more interesting. They are classical analogs of boundary Knizhnik-Zamolodchikov-Bernard  differential operators \cite{SR}\cite{RS}. We have the following formula for $D_k$:
\[
D_k=(\mu^{(k)}_0,p)-\sum_{l=1}^{k-1}(r_{lk}+r_{lk}^{\theta_l})+(\sum_\alpha K_\alpha\mu^{(k)}_{-\alpha}-\kappa_k)
+\sum_{l=k+1}^n(r_{kl}-r_{kl}^{\theta_k}).
\]
%\[
%D_k=(\mu^{(k)}_0,p)-\kappa_k-\sum_{l<k}(r_{lk}+r_{lk}^{\theta_l})-\sum_{l>k}(-r_{kl}+r_{lk}^{\theta_l})+\sum_\alpha \mu^{(k)}_{-\alpha} K_\alpha.
%\]
Here $r_{kl}$ for $k\not=l$ now is Felder dynamical $r$-matrix rescaled in $a\in A_{reg}$,
\begin{equation}\label{F-r-b}
r_{kl}=-\frac{1}{2}(\mu^{(k)}_0, \mu^{(l)}_0)+\sum_{\alpha}\frac{ \mu^{(k)}_{-\alpha}\mu^{(l)}_{\alpha}}{a_\alpha^2-1},
\end{equation}
$\theta_k$ is the transpose of the Cartan involution acting on $\mu^{(k)}$,
\[
\kappa_k=\frac{1}{2} (\mu^{(k)}_0,\mu^{(k)}_0)+\sum_\alpha\frac{(\mu^{(k)}_\alpha)^2}{1-a_{\alpha}^2}
\]
and 
\begin{equation}\label{Ka}
K_\alpha=\frac{a_\alpha\mu'_{[\alpha]}+\mu''_{[\alpha]}}{a_\alpha-a_{\alpha}^{-1}}.
\end{equation}
%\begin{equation}\label{Ka}
%K_\alpha=\frac{\mu''_\alpha-a_{-\alpha}\mu'_\alpha}{a_\alpha-a_{-\alpha}}.
%\end{equation}
The differences $D_k=H_2^{(k)}-H_2^{(k-1)}$ are classical analogs of boundary KZB operators derived in \cite{SR}\cite{RS}. The superintegrability of open spin CM chains is proven in section \ref{op-sint}. The projection method for solving equations of motion and angle coordinates are described in section \ref{op-proj-dyn}.
  
\subsection*{6} The structure of the paper is as follows. In section \ref{pCM} we construct periodic
spin CM chains by the Hamiltonian reduction and prove the superintegrability. In section \ref{g-act-per}
we describe the phase space of such a system. In sections \ref{g-red}, \ref{ph-space-reg} we describe the regular part of the phase space. Hamiltonians of a periodic spin CM chain, restricted to the regular 
part of the phase space are described in section \ref{Ham-per}. The superintegrability of a periodic
spin CM chain is proven in section \ref{pCM-sup}. In section \ref{solutions} solutions to equations of motion are described algebraically by the projection method, and angle variables are described. In section \ref{open-CM} we focus on open spin CM chains. In section \ref{ph-space-open}
we describe phase spaces. In section \ref{op-g-reduct}, \ref{op-s-leaf} we describe the regular part 
of the phase space. Hamiltonians of an open spin CM chain, restricted to the regular 
part of the phase space are described in section 
%\ref{oCM-Ham} and 
\ref{qu-Ham}. 
The superintegrability of an open spin CM chain is proven in section \ref{op-sint}. In section  \ref{op-proj-dyn} solutions to equations of motion are described algebraically
% in section \ref{op-proj-dyn} 
by the projection method, and angle variables are described.
 %in section \ref{angles}.  
 In the conclusion (section \ref{concl}) we discuss some open problems and
describe in details periodic CM spin chain for $SL_N$ with orbits of rank 1.
In Appendix 
%\ref{z-Poisson} we compute explicitly intermediate Poisson for the
%Hamiltonian reduction in the periodic case and in Appendix 
\ref{R3} we compare 
our symplectic leaves with the ones from \cite{R3}.

Throughout this paper we will focus on split real semisimple Lie groups. However, since all constructions are algebraic they extend (with appropriate modifications) to the complex algebraic case. The non-split real case will be the subject of a separate publication (see \cite{RS} for the quantum case). 
Another important real case is when $G$ is compact, 
%and $K=G^{\Theta}$ is the fix-point subgroup of an involution on $G$, 
which can be deduced from the complex algebraic case by restriction to a compact real form.
The structure of phase spaces as stratified symplectic spaces will be explored further in \cite{CJRX}.
%We will briefly discuss corresponding integrable systems here.

%This paper is a short version of an upcoming paper where we are
%connecting these models with superintegrable systems on moduli spaces
%of flat connections and their generalizations. Some of the proof are left
%to a longer version of this paper.

\subsection*{Acknowledgments}This paper was started as a joint project with Jasper Stokman. The author is grateful to Jasper for many discussions and for the collaboration on this paper. He also would like to thank Vladimir Fock, Eva Miranda and Hessel Postuma for important discussions and remarks and to Zhuo Chen, Kai Jiang and Husileng Xiao  for
discussions on stratified symplectic spaces. N.R. want to thank ITS-ETH for the hospitality, where the bulk of this work was completed. The work of N.R. was supported by the NSF grant DMS-1902226, by the Dutch research council (NWO 613.009.126), and by the grant RFBR No. 18-01-00916. 

%%%%%%%%%%%%%%%%%%%%%%%%%%%%%%%%%%%%%%%%%%%
\section{Periodic spin Calogero-Moser chains}\label{pCM}
%%%%%%%%%%%%%%%%%%%%%%%%%%%%%%%%%%%%%%%%%%
\subsection{The phase space as the Hamiltonian reduction} \label{g-act-per}
%%%%%%%%%%%%%%%%%%%%%%%%%%%%%%%%%%%%%%%%%%
Here we will describe the phase space of a periodic spin Calogero-Moser chain as a Hamiltonian reduction of $T^* (G^{\times n})$. Let us start with the description of these symplectic spaces.

Consider the manifold $T^*(G^{\times n})$ with the standard symplectic structure.
The cotangent bundle over a Lie group can be trivialized by right translations, which gives an isomorphism of vector bundles
\[
T^*(G^{\times n})\simeq (T^*G)^{\times n}\simeq {\g^*}^{\times n}\times G^{\times n}
\]
We will choose this trivialization throughout the paper.

The Lie group $G_n:=G^{\times n}$ acts naturally on itself by left and right translations. 
Lifting these actions to $T^*(G^{\times n})$, after the trivialization of the cotangent bundle, 
we can write the action by left translations as:
\[
h_L(x,g)=(Ad_{h_1}^*(x_1), Ad_{h_2}^*(x_2) \dots, Ad_{h_{n}}^*(x_n), h_1g_1, h_2g_2, \dots, h_{n}g_n)
\]
and the action by right translations as
\[
h_R(x,g)=(x_1,\dots, x_n, g_1h_1^{-1}, \dots, g_nh_n^{-1})
\]
Both these actions are Hamiltonian with moment maps
\[
\mu_L(x,g)=(x_1,x_2,\dots, x_n)
\]
and
\[
\mu_R(x,g)=(-Ad_{g_1^{-1}}^*(x_1), \dots, -Ad_{g_n^{-1}}^*(x_n))
\]
respectively.

Actions by left and right translations can be twisted by permutations. In particular, we can twist the action by left translations by a cyclic permutation.
Combining the twisted left action with the non-twisted right action we obtain the "gauge action" of $G_n$ on $G^{\times n}${\footnote{One can twist both left and right actions
by a permutation. This leads to other superintegrable systems.}
\[
h(g_1,\dots, g_n)=(h_1g_1h_2^{-1},h_2g_2h_3^{-1},\dots, h_{n}g_nh_1^{-1})
\]

Lifting the twisted conjugation action of $G_n$ on $G^{\times n}$ to $T^*(G^{\times n})$ we obtain the "gauge action" on the cotangent bundle:
\begin{equation}\label{tad-act}
h(x,g)=(Ad_{h_1}^*(x_1), Ad_{h_2}^*(x_2), \dots, Ad_{h_{n}}^*(x_n), h_1g_1h_2^{-1}, h_2g_2h_3^{-1}, \dots, h_{n}g_nh_1^{-1})
\end{equation}
Because  this is the diagonal action for two Hamiltonian actions, the gauge action is also Hamiltonian with the moment map $\mu: T^*(G^{\times n})\to {\g^*}^{\times n}$:

\begin{equation}\label{tad-mmap}
\mu(x,g)=\mu_L(x,g)+\mu_R^{tw}(x,g)=(x_1-Ad_{g_n^{-1}}^*(x_n), x_2-Ad_{g_1^{-1}}^*(x_1), \dots, x_n-Ad_{g_{n-1}^{-1}}^*(x_{n-1}))
\end{equation}
where $\mu_R^{tw}$ is the right
moment map, twised by cyclic permutation. 

Because the gauge action \eqref{tad-act} of $G_n$
is Hamiltonian, the quotient space $T^*(G^{\times n})/G_n$
%with respect to the action (\ref{tad-act}) of $G^{\times n}$ 
 is a Poisson space.\footnote{This space is singular. Having in mind classical-quantum correspondence we need the algebra of functions on $T^* G^{\times n}/G_n$. Thus, by the quotient space we will always mean the GIT quotient. By definition, functions on $T^*(G^{\times n})/G_n$ are $G_n$-invariant functions on $T^*(G^{\times n})$.} 
Symplectic leaves of $T^*(G^{\times n})/G_n$
are given by the Hamiltonian reduction with respect to the moment map (\ref{tad-mmap}).
Let $\cO_1,\dots, \cO_n$ be coadjoint orbits in $\g^*$, then the corresponding symplectic leaf 
in $T^*(G^{\times n})/G_n$
is
\begin{equation}\label{SOoriginal}
\cS(\cO)=\mu^{-1}(\cO_1\times\dots\times \cO_n)/G_n=\{(x,g)\in {\g^*}^{\times n}\times G^{\times n}| x_{i}-Ad_{g_{i-1}^{-1}}^*(x_{i-1})\in \cO_i\}/G_n
%\mu^{-1}(\cO_1\times\dots\times \cO_n)/G^{\times n}=\{(x,g)\in {\g^*}^{\times n}\times G^{\times n}| x_{i+1}-Ad_{g_i^{-1}}^*(x_i)\in \cO_i\}/G^{\times n}
\end{equation}
where $G_n$ acts by the gauge transformations (\ref{tad-act}) and the indices $i$ should be taken modulo $n$.

On each of these symplectic leaves we will construct a superintegrable system which we will call a {\it periodic spin Calogero-Moser chain}.

\subsection{The gauge fixing}\label{g-red} 
Let us fix $i\in1,\dots, n$ and $g=(g_1,\ldots,g_n)\in G^{\times n}$. Let $h\in G_n$ such that 
\begin{equation*}
h_j=
\begin{cases}
h_ig_{i-1}^{-1}\cdots g_{j+1}^{-1}g_j^{-1}\quad &\hbox{ for }\,\, 1\leq j<i,\\
h_ig_{i-1}^{-1}\cdots g_2^{-1}g_1^{-1}g_n^{-1}\cdots g_{j+1}^{-1}g_j^{-1}\quad &\hbox{ for }\,\, i<j\leq n.
\end{cases}
\end{equation*}
%\begin{eqnarray}\label{gauge}
%h_j=h_{i+1}g_{i+1}g_{i+2}\dots g_{j-2}g_{j-1}
%%h_1&=&h_ig_{i+1}\dots g_Ng_1, h_2=h_ig_{i+1}\dots g_Ng_1g_2, \dots, h_{i-1}=h_ig_{i+1}\dots g_Ng_1
%%\dots g_{i-1}, \\ \nonumber h_{i+1}&=&h_ig_{i+1}, \dots h_N=h_ig_{i+1}\dots g_N,  
%\end{eqnarray} 
%for each $j\in 1,\dots, n$. Here the product is cyclic. 
Denote such element of $G_n$  by $h_g$ (we suppress the dependence on $i$).

It is easy to check that the gauge transformation of $g=(g_1,\dots, g_n)$ 
by the element $h_g$ brings it to $(1,\ldots,1,h_i(g_ig_{i+1}\cdots g_ng_1g_2,\ldots g_{i-1})h_i^{-1},1,\ldots,1)$, with the $i^{\textup{th}}$-entry being the nontrivial entry.
%$g^{h_g}=(1, \dots, h_{i+1}(g_{i+1}\dots g_ng_1g_2\dots g_{i-1}g_i)h_{i+1}^{-1},\dots, 1)$. 
This identifies the $G_n$ gauge orbit through $g=(g_1,\dots, g_n)$ with the $G$-conjugation orbit through $g_1\cdots g_n$. It thus gives an ($i$-independent) isomorphism
\[
G^{\times n}/G_n\overset{\sim}{\longrightarrow} G/G,
\]
where $G/G$ denotes the set of conjugacy classes in $G$. On the cotangent bundles the gauge fixing with gives the isomorphism
$\varphi_i: \bigl(\mathfrak{g}^{*\times n}\times G^{\times n}\bigr)/G_n\overset{\sim}{\longrightarrow}\bigl(\mathfrak{g}^{*\times n}\times G\bigr)/G$ mapping the $G_n$-orbit
$G_n(x,g)$ through $(x,g)\in\mathfrak{g}^{*\times n}\times G^{\times n}$ to the $G$-orbit through 
\[
\bigl(Ad_{g_i^{-1}\cdots g_1^{-1}}^*(x_1),\ldots,Ad_{g_i^{-1}}^*(x_i),Ad_{g_i^{-1}\cdots g_1^{-1}g_n^{-1}\cdots g_{i+1}^{-1}}^*(x_{i+1}),
\ldots,Ad_{g_i^{-1}\cdots g_1^{-1}g_n^{-1}}^*(x_n),g_{i+1}\cdots g_ng_1\cdots g_i\bigr)
\]
(for $i=n$ this should be read as $\bigl(Ad_{g_n^{-1}\cdots g_1^{-1}}^*(x_1),\ldots,Ad_{g_n^{-1}}^*(x_n),g_1\cdots g_n\bigr)$).
Here $G$ is acting diagonally on $\mathfrak{g}^{*\times n}\times G$ via the coadjoint action on $\g^*$ and the conjugation action on $G$. From now on we will work with the isomorphism $\varphi_n$.
%representatives
%of $G_n$-orbits in $G^{\times n}/G_n$ with all coordinates equal to $1$, except for the last coordinate (i.e., we take $i=n$).
%Is easy to prove that this is an isomorphism. 

%We can chose $h_{i+1}$ in such a way that 
%$h_{i+1}(g_{i+1}\dots g_ng_1g_2\dots g_{i-1}g_i)h_{i+1}^{-1}\in H$ is an 
%element of the Cartan subgroup in $G$. This gives an isomorphism $\varphi_i$:
%\[
%G^N/G^N\simeq G/G\simeq H/W
%\]
%for each $i=1,\dots, n$.

\subsection{The regular part of the phase space} \label{ph-space-reg}

%From now on we consider $\varphi_n$. 
%The isomorphism $\varphi$ extends to an isomorphism
%$\varphi: \bigl(\mathfrak{g}^{*\times n}\times G^{\times n}\bigr)/G_n\overset{\sim}{\longrightarrow}\bigl(\mathfrak{g}^{*\times n}\times G\bigr)/G$
%mapping the $G_n$-orbit through $(x,g)\in\mathfrak{g}^{*\times n}\times G^{\times n}$ to the $G$-orbit through $\bigl(\textup{Ad}_{g_n^{-1}\cdots g_1^{-1}}^*(x_1),\ldots,\textup{Ad}_{g_n^{-1}}^*(x_n),g_1\cdots g_n\bigr)$ (this lift of the isomorphism $\varphi$ depends on the choice . Here $G$ is acting diagonally on $\mathfrak{g}^{*\times n}\times G$ via the coadjoint action on $\g^*$ and the conjugation action on $G$.

The image of the symplectic leaf $\cS(\cO)$ under the isomorphism $\varphi_n$  is
\[
\cS(\cO)=\{(z_1,\dots, z_n, g)\in \g^{*\times n} \times
G\,\, |\,\, z_1-Ad^*_{g^{-1}}z_n\in \cO_1,\,\, z_{i}-z_{i-1}\in \cO_{i},\,\,\, i=2,\dots,n\}/G.
%\cS(\cO)=\{(z_1,\dots, z_n, g)\in \g^{*\times n} \times
%%\dots \times \g^*\times 
%G| z_{i}-z_{i+1}\in \cO_i, \ \ i=1,\dots, n-1, \ \ z_1-Ad^*_{g^{-1}}z_z\in \cO_n\}/G
\]
Define the {\it regular} part $\cS(\cO)_{reg}\subset \cS(\cO)$ of the phase space as $\cS(\cO)\cap (\mathfrak{g}^{*\times n}\times G^\prime)/G$.
% the set of $G$-orbits in $\cS(\cO)$ with last coordinate $g$ in $G^\prime$.
%$g\in G_{reg}$\footnote{As usual we will say that $g\in G$ is regular if it can be brought by a conjugation to a regular element of the Cartan subgroup. An element of the Cartan subgroup is called regular if it has a trivial stabilizer in the Weyl group. We will denote the space of regular elements in $G$ by $G_{reg}$. The set of regular elements in the Cartan subgroup $H$ we denote by $H_{reg}$.}
On $\cS(\cO)_{reg}$ we can choose a representative where $g$ is in the regular part $A_{reg}$ of the real split torus $A$ in $G$:
%the Cartan subgroup $H\subset G$. 
$g=bab^{-1}, z_i=\textup{Ad}_b^*x^{(i)}$ with $a\in A_{reg}$. 
Then 
we have
\begin{equation*}
\begin{split}
\cS(\cO)_{reg}=\{(x^{(1)},\dots, x^{(n)}, a)\in \g^{*\times n} \times A_{reg}\,\,|\,\, &x^{(1)}-Ad^*_{a^{-1}}x^{(n)}\in \cO_1,\\
&\,\,\,x^{(i)}-x^{(i-1)}\in\cO_i,\,\,  i=2,\dots, n\}/N_G(A).
\end{split}
\end{equation*}
%\end{align*}
%where $N_G(A)\subset G$ is the normalizer of $A$ in $G$. 
%Choose a basis $e_\alpha$ in each root subspace $\g_\alpha\subset \g$. 

Identify $\mathfrak{g}^*\simeq\mathfrak{g}$ and $\mathfrak{a}^*\simeq\mathfrak{a}$ via the Killing form of $\g$. The element $y\in\mathfrak{g}^*$ then corresponds to $y_0+\sum_\alpha y_\alpha e_\alpha$, where $y_0$ is the element in $\mathfrak{a}$ corresponding to $y\vert_{\mathfrak{a}}$ and 
$y_\alpha=y(e_{-\alpha})$. 
%$x_\alpha$ are corresponding coordinate functions.
Let $\mu^{(j)}\in \cO_j$ be vectors $\mu^{(1)}=x^{(1)}-Ad^*_{a^{-1}} x^{(n)}$ and
$\mu^{(i)}=x^{(i)}-x^{(i-1)}$ for $i=2,\dots, n$. For coordinates  $x^{(i)}_\alpha$ and $\mu^{(i)}_\alpha$  of vectors $x^{(i)}$ and $\mu^{(i)}$ we then have
%in the root basis $e_\alpha$ we then have
\[
x_\alpha^{(1)}-a_\alpha^{-1} x^{(n)}_\alpha=\mu^{(1)}_\alpha,\qquad\quad x^{(i)}_\alpha-x^{(i-1)}_\alpha=\mu^{(i)}_\alpha, \qquad i=2,\dots, n.
\]
For the Cartan components we have
\[
x^{(i)}_0-x^{(i-1)}_0=\mu^{(i)}_0, \qquad i=1,\dots, n,
\]
with the index $i$ taken to be modulo $n$.
 
Solving these equations for $x^{(i)}$ we have
\begin{equation}\label{solve}
\begin{split}
x^{(i)}_\alpha&=\frac{a_\alpha(\mu_\alpha^{(1)}+\mu_\alpha^{(2)}+\cdots+\mu_\alpha^{(i)})+\mu_\alpha^{(i+1)}+\mu^{(i+2)}+\cdots+\mu_\alpha^{(n)}}{a_\alpha-1},\\
%\[
%x^{(k)}_\alpha=\frac{h_\alpha(\mu^{(n)}_\alpha+\mu^{(1)}_\alpha +\dots +\mu^{(k-1)}_\alpha)+\mu^{(k)}_\alpha +\dots +\mu^{(n-1)}_\alpha}{h_\alpha-1}
%\]
x^{(i)}_0&=x^{(1)}_0+ \mu^{(2)}_0 +\dots +\mu^{(i)}_0=x^{(n)}_0-\mu^{(n)}_0 -\dots -\mu^{(i+1)}_0
\end{split}
\end{equation}
%Here the last sum is cyclic and 
and we have the constraint
\begin{equation}\label{constraint}
\mu^{(1)}_0 +\dots +\mu^{(n)}_0=0.
\end{equation}

This gives an isomorphism 
\begin{equation}\label{radred}
\cS(\cO)_{reg}\overset{\sim}{\longrightarrow} \bigl(\nu_{\cO}^{-1}(0)/H\times T^*A_{reg}\bigr)/W
%((\cO_1\times \dots \times \cO_n)\sslash A \times T^*A_{reg})/W,
\end{equation}
%\[
%\cS(\cO)_{reg}\simeq ((\cO_1\times \dots \times \cO_n)\sslash A \times T^*A_{reg})/W,
%\]
which preserves the natural symplectic structures,
where $\nu_\cO: \cO_1\times\cdots\times\cO_n\rightarrow\mathfrak{a}^*$ is the moment map $(\mu^{(1)},\ldots,\mu^{(n)})\mapsto(\mu^{(1)}+\cdots+\mu^{(n)})\vert_{\mathfrak{a}}$ for the diagonal action of $H$ on the product $\cO_1\times\cdots\times\cO_n$ of coadjoint orbits, and $W=N_G(A)/H$ acts diagonally on $\nu_{\cO}^{-1}(0)/H\times T^*A_{reg}$.
%$(\cO_1\times \dots \times \cO_n)\sslash A=\{( \mu^{(1)}, \dots , \mu^{(n)})\in \cO_1\times \dots \times \cO_n) \,\,|\,\,\mu^{(1)}_0 +\dots +\mu^{(n)}_0=0\}$ 
%is the Hamiltonian reduction of $\cO_1\times \dots \times \cO_n$ with respect to the diagonal $A$-action.
%of the Cartan subgroup $H\subset G$. 
The isomorphism \eqref{radred} maps the $N_G(A)$-orbit through $\bigl(x^{(1)},\ldots,x^{(n)},a\bigr)$ to the $W$-orbit through $\bigl(H(x^{(1)}-Ad_{a^{-1}}^*x^{(n)},x^{(2)}-x^{(1)},\ldots,x^{(n)}-x^{(n-1)}),x_0^{(n)},a\bigr)$,
where we used the trivialisation 
$T^*A_{reg}\simeq\mathfrak{a}\times A_{reg}$. The inverse maps the $W$-orbit through $\bigl(H(\mu^{(1)},\ldots,\mu^{(n)}),p,a\bigr)$ to the $N_G(A)$-orbit through
$(x^{(1)},\ldots,x^{(n)},a)$, with
\begin{equation}\label{relationxmu}
x^{(i)}=p-\mu_0^{(n)}-\cdots-\mu_0^{(i+1)}+\sum_\alpha\left(\frac{a_\alpha(\mu_\alpha^{(1)}+\cdots+\mu_\alpha^{(i)})+\mu_\alpha^{(i+1)}+\cdots+\mu_\alpha^{(n)}}
{a_\alpha-1}\right)e_\alpha,
\end{equation}
where we use the identification $\g\simeq\g^*$ via the Killing form.
%$(x^{(n)}_0, h)$ to $T^*H_{reg}$.
%The action of the Weyl group on the space $(\cO_1\times \dots \times \cO_n)\sslash H$ descends from the action of the normalizer $N(H_{reg})$ on $\cO_1\times \dots \times \cO_n$ by the diagonal coadjoint action.

%Note that the ring of algebraic functions on $\cS(\cO)_{reg}$ is the $W$-invariant part of the quotient ring of $H$-invariant polynomials in $\mu_\alpha^{(i)}, (\mu^{(i)}_0)_a, (x^{(n)}_0)_a$ and rational functions in $h_\alpha^{\pm 1}-1 $ where $a=1,\dots, r$, $\alpha\in \Delta$, with respect to the ideal generated by $\sum_{i=1}^n (\mu^{(i)}_0)_a$. 

\subsection{Hamiltonians of a periodic spin CM chain}\label{Ham-per}
After the trivialization of the cotangent bundle by translations, we have a natural projection:
\begin{equation}\label{pr-1}
T^*(G^{\times n})\simeq {\g^*}^{\times n}\times G^{\times n}\to {\g^*}^{\times n}
\end{equation}
which is simply the projection to the first factor. 
This projection depends on the trivialization. In this paper we alway assume that we use the trivialization by right translations.
However, the corresponding projection of quotient spaces 
\begin{equation}\label{pr-2}
T^*(G^{\times n})/G_n\to ({\g^*/G})^{\times n}
\end{equation}
does not depend on the trivialization and in this sense is canonical.

The projection (\ref{pr-2}) is Poisson\footnote{One of the reasons for this is that the equation (\ref{pr-1}) is the moment map for
the left diagonal action of $G^{\times n}$ on the cotangent bundle.} with the trivial Poisson structure on ${(\g^*/G)}^{\times n}$.
Thus the $G^{\times n}$-invariant functions on ${\g^*}^{\times n}$ give rise to
% form 
a Poisson commutative subalgebra 
in the algebra of functions on $T^*(G^{\times n})/G_n$. The restriction of these functions to the symplectic leaf $\cS(\cO)$ gives the algebra of Poisson commuting functions on it.
This is the subalgebra of {\it Hamiltonians of the periodic spin Calogero-Moser chain}.

Now let us describe the restriction of the Hamiltonians corresponding to quadratic Casimir functions  
\[
H_2^{(k)}(x,g)=\frac{1}{2}\bigl(x^{(k)},x^{(k)}\bigr)=\frac{1}{2}\bigl(x_0^{(k)},x_0^{(k)}\bigr)+\sum_{\alpha>0}x^{(k)}_\alpha x^{(k)}_{-\alpha}\qquad
(1\leq k\leq n)
%=\frac{1}{2}(x_k,x_k)=\frac{1}{2}(z_k,z_k)=\frac{1}{2}(x^{(k)},x^{(k)})=\frac{1}{2}(x_0^{(k)},x_0^{(k)})+\sum_{\alpha>0}x^{(k)}_\alpha x^{(k)}_{-\alpha}
\]
to the regular part of $\cS(\cO)$, where $x=(x^{(1)},\ldots,x^{(n)})\in\g^{*\times n}$ and $g\in G^{\times n}$. Consider the functions
\[
D_k=H_{2}^{(k)}-H_2^{(k-1)}\qquad\quad (1<k\leq n)
\]
which we call {\it Knizhnik-Zamolodchikov-Bernard (KZB) Hamiltonians}\footnote{The proper name would be {\it constant Knizhnik-Zamolodchikov-Bernard Hamiltonians} emphasizing  the fact that they are related to finite dimensional simple Lie algebras, not to the affine Kac-Moody algebras. See for example references \cite{F}\cite{ES}\cite{EV}\cite{S}.}. 

\begin{theorem}\label{thmradcyclic}  The restriction of the KZB Hamiltonians to $\cS(\cO)_{reg}$ can be written as
\begin{equation}\label{Dk-per}
D_k=(\mu^{(k)}_0, p)-\sum_{l=1}^{k-1} r_{lk}+\sum_{l=k+1}^{n}r_{kl}
\end{equation}
where $r_{kl}$ for $k\not=l$ is the classical version of the Felder's dynamical $r$-matrix \cite{F}:
\begin{equation}\label{F-r}
r_{kl}=-\frac{1}{2}(\mu^{(k)}_0, \mu^{(l)}_0)+\sum_{\alpha}\frac{ \mu^{(k)}_{-\alpha}\mu^{(l)}_{\alpha}}{a_\alpha-1}.
\end{equation}
\end{theorem}
\begin{remark}
Note that (\ref{F-r}) can also be written as
\[
r_{kl}=-\frac{1}{2}(\mu^{(k)}_0, \mu^{(l)}_0)+\sum_{\alpha>0}\frac{ \mu^{(k)}_{-\alpha}\mu^{(l)}_{\alpha}}{a_\alpha-1}-\sum_{\alpha>0}\frac{a_\alpha\mu^{(k)}_{\alpha} \mu^{(l)}_{-\alpha}}{a_\alpha-1}.
\]
\end{remark}
\begin{proof}
We need to show that formula \eqref{Dk-per} gives the expression of $D_k$ in terms of the coordinates on $\cS(\cO)_{reg}$ of $\bigl(\nu_{\cO}^{-1}(0)/H\times T^*A_{reg}\bigr)/W$, obtained from the isomorphism \eqref{radred}. In particular, let
$\bigl(H(\mu^{(1)},\ldots,\mu^{(n)}),p,a\bigr)\in \nu_{\cO}^{-1}(0)/H\times\mathfrak{a}^*\times A_{reg}$ and let $\bigl((x^{(1)},\ldots,x^{(n)}),(1,\ldots,1,a)\bigr)$ be the corresponding point in
$\g^{*\times n}\times G^{\times n}$, with $x^{(i)}$ given by \eqref{relationxmu}.
Taking into account the relation $x^{(k)}-x^{(k-1)}=\mu^{(k)}$ between the $x^{(i)}$ and the $\mu^{(j)}$ we have 
\begin{equation*}
\begin{split}
D_k&=\bigl(\mu^{(k)},x^{(k-1)}\bigr)+\frac{1}{2}\bigl(\mu^{(k)},\mu^{(k)}\bigr)\\
&=\bigl(\mu_0^{(k)},x_0^{(k-1)}+\frac{1}{2}\mu_0^{(k)}\bigr)+
\sum_\alpha x_\alpha^{(k-1)}\mu_{-\alpha}^{(k)}+\sum_{\alpha>0}\mu_\alpha^{(k)}\mu_{-\alpha}^{(k)}.
%\frac{1}{2}(\mu^{(k)}_0, x^{(k+1)}_0)+x^{(k)}_0)+\sum_\alpha x^{(k)}_\alpha\mu^{(k)}_{-\alpha}+
%\sum_{\alpha>0} \mu^{(k)}_{\alpha}\mu^{(k)}_{-\alpha}
\end{split}
\end{equation*}
Substitute here the expression \eqref{solve} 
%\eqref{relationxmu} 
for $x^{(k-1)}_\alpha$ in terms of the $\mu_\alpha^{(j)}$:
\[
D_k=\bigl(\mu_0^{(k)},x_0^{(k-1)}+\frac{1}{2}\mu_0^{(k)}\bigr)
%\frac{1}{2}(\mu^{(k)}_0, x^{(k+1)}_0)+x^{(k)}_0)
+\sum_{l=1}^{k-1} \sum_\alpha \frac{a_\alpha\mu_\alpha^{(l)}\mu^{(k)}_{-\alpha}}{a_{\alpha}-1}+
\sum_{l=k+1}^{n} \sum_\alpha \frac{\mu_\alpha^{(l)}\mu^{(k)}_{-\alpha}}{a_{\alpha}-1}.
\]
From here using the identities
\[
\sum_\alpha \frac{a_\alpha\mu_\alpha^{(l)}\mu^{(k)}_{-\alpha}}{a_{\alpha}-1}=-r_{lk}-\frac{1}{2}(\mu_0^{(k)},\mu_0^{(l)})
\]
and 
\[
\sum_\alpha \frac{\mu_\alpha^{(l)}\mu^{(k)}_{-\alpha}}{a_{\alpha}-1}=r_{kl}+\frac{1}{2}(\mu_0^{(k)},\mu_0^{(l)})
\]
we conclude 
\begin{equation*}
D_k=\Bigl(\mu^{(k)}_0,x_0^{(k-1)}+\frac{1}{2}\sum_{l=k}^{n}\mu_0^{(l)}-\frac{1}{2}\sum_{l=1}^{k-1}\mu_0^{(l)}\Bigr)
%x^{(k+1)}_0+x^{(k)}_0-\mu^{(n)}_0-\mu^{(1)}_0-\dots - \mu^{(k-1)}_0+ \mu^{(k+1)}_0+\dots +  \mu^{(n-1)}_0)\\
-\sum_{l=1}^{k-1}r_{lk}+\sum_{l=k+1}^{n}r_{kl}.
\end{equation*}
Using $x_0^{(k-1)}=p-\mu_0^{(n)}-\cdots-\mu_0^{(k)}$ (see \eqref{relationxmu}) and the constraint \eqref{constraint}
%Taking into account relations between $x$ and $\mu$ and setting $x^{(n)}_0=p$ 
we obtain (\ref{Dk-per}).
\end{proof}

A particularly simple expression has the quadratic Hamiltonian $H_2^{(n)}$ on $\cS(\cO)_{reg}$,
\[
H_2^{(n)}=\frac{1}{2} (p,p)+\sum_{\alpha}\frac{\mu_\alpha \mu_{-\alpha}}{(1-a_\alpha)(1-a_{\alpha}^{-1})}.
\]
Here $\mu_\alpha=\mu^{(1)}_\alpha+\dots +\mu^{(n)}_\alpha$.
Setting 
%$p=x^{(n)}_0$ and 
$q_\alpha=\alpha(\log(a))$
%$a_\alpha=\exp(q_\alpha)$ 
this formula becomes
a familiar formula for the spin Calogero-Moser Hamiltonian,
\begin{equation}\label{Hamper}
H_2^{(n)}=\frac{1}{2} (p,p)-\sum_{\alpha>0}\frac{\mu_\alpha \mu_{-\alpha}}{2sh^2(q_\alpha)}.
\end{equation}

Note that the periodic spin CM chain is the classical version of the dynamical Knizhnik-Zamolodchikov equation from \cite{ES}\cite{EV}.

\subsection{Periodic spin Calogero-Moser chain as a superintegrable system}\label{pCM-sup}

Now let us establish the superintegrability of the periodic spin CM chain. For this we should construct 
an intermediate Poisson manifold and projections as in \cite{N}\cite{R2}.

%\subsubsection{} 
Observe that we have natural Poisson projections:
\begin{equation}\label{pr-3}
T^*(G^{\times n})/G_n \stackrel{p_1}{\rightarrow} \cP_n \stackrel{p_2}{\rightarrow} \cB_n
\end{equation}

Firstly, $\cP_n=({\g^*}^{\times n}\times_{{(\g^*/G)}^{\times n}}{\g^*}^{\times n})/G_n$ with
\begin{equation*}
{\g^*}^{\times n}\times_{{(\g^*/G)}^{\times n}}{\g^*}^{\times n}:=\{(x,y)\in {\g^*}^{\times n}\times {\g^*}^{\times n}| Gy_{i}=-Gx_{i-1}\},
\end{equation*}
where $Gz$ is the coadjoint orbit through $z\in\g^*$ and the indices $i$ are taken modulo $n$, and $G_n$ is acting by
\begin{equation}\label{P2}
g(x,y):=(Ad_{g_1}^*(x_1),\ldots,Ad_{g_n}^*(x_n),Ad_{g_1}^*(y_1),\ldots,Ad_{g_n}^*(x_n)).
\end{equation}
The map $p_1$ is the map induced from the $G_n$-equivariant map $\mu_L\times \mu_R^{tw}$. Explicitly, the mapping $p_1$ acts as
\begin{equation}\label{P1}
\begin{split}
p_1: G_n(x,g)&\mapsto G_n(\mu_L(x,g),\mu^{tw}_R(x,g))\\
&=G_n(x_1,x_2,\ldots,x_n, -Ad_{g_n^{-1}}^*(x_n),-Ad_{g_1^{-1}}(x_1),\ldots,-Ad_{g_{n-1}^{-1}}(x_{n-1})).
% \dots, -Ad_{g_n^{-1}}^*(x_n), -Ad_{g_1}^*(x_1))
\end{split}
\end{equation}
%\[
%\cP_n=({\g^*}^{\times n}\widetilde{\times}_{{(\g^*/G)}^{\times n}}{\g^*}^{\times n})/G^{\times n}=\{(x,y)\in {\g^*}^{\times n}\times {\g^*}^{\times n}| Gx_{i+1}=-Gy_{i}\}/G^{\times n}
%\]

Secondly,
\[
\cB_n={(\g^*/G)}^{\times n}
\]
and the map $p_2$ is the projection to the first factor. 

%We used notations:
%\[
%{\g^*}^{\times n}\widetilde{\times}_{{(\g^*/G)}^{\times n}}{\g^*}^{\times n}=\{(x,y)\in {\g^*}^{\times n}\times {\g^*}^{\times n}| Gx_{i+1}=-Gy_{i}\}
%\]
%\end{eqnarray}
%were $G^{\times n}$ acts as
%\begin{eqnarray}\label{P2}
%(g_1,\dots, g_n)(x_1, \dots, x_n, y_1, \dots, y_n)= \\ \nonumber (Ad^*_{g_1}(x_1), \dots, Ad^*_{g_n}(x_n), Ad^*_{g_2}(x_1), \dots, Ad^*_{g_n}(x_{n-1}), Ad^*_{g_1}(x_n)
%\end{eqnarray}
%The mapping $p_2$ is 
%\begin{align*}
%\pi_2: (G_n\times G_n)(x,y)=(Gx_1, Gx_2, \dots, Gx_n, Gy_1, Gy_2, \dots, Gy_n)=\\
%(Gx_1, Gx_2, \dots, Gx_n, -Gx_n, -Gx_1, \dots, -Gx_{n-1})\mapsto (Gx_1, \dots, Gx_n)
%\end{align*}

%The isomorphism $\cB_n\simeq {(\g^*/G)}^{\times n}$ is the projection to the first $n$ factors
%in ${\g^*}^{\times 2n}$.

%\subsubsection{} 
Restricting projection $p_1$ to the symplectic leaf $\cS(\cO)$ (see \eqref{SOoriginal}), we obtain the surjective Poisson projection
\[
p_{1,\cO}: \cS(\cO)\to \cP(\cO)
\]
where 
\[
\cP(\cO)=\{(x,y)\in {\g^*}^{\times n}\times_{(\g^*/G)^{\times n}}{\g^*}^{\times n}\,\,|\,\,x_i+y_i\in \cO_i\}/G_n\subset\cP_n
\]
with the $G_n$-action described by \eqref{P2}.
%\[
%\cP(\cO)=\{(x,y)\in {\g^*}^{\times n}\times {\g^*}^{\times n}|x_i+y_i\in \cO_i, \ \ Gx_{i+1}=-Gy_i\}/G^{\times n}
%\]
%Here the action of $G^{\times n}$ on $x$ and $y$ and the projection $p_1$ are described in (\ref{P1})(\ref{P2}) 

Restricting the second projection $p_2$ to $\cP(\cO)$ we have the Poisson projection 
\[
p_{2,\cO}: \cP(\cO)\to \cB(\cO)\subset \cB_n,\qquad
%{(\g^*/G)}^{\times n}, 
G_n(x,y)\mapsto (Gx_1, \dots, Gx_n)
\]
where $\cB(\cO)$ is the image of $p_{2,\cO}$. It can be explicitly described as
\[
\cB(\cO)=\{(\cO^{(1)}, \ldots, \cO^{(n)})\in (\g^*/G)^{\times n}\,\,|\,\, \cO_i\subseteq\cO^{(i)}-\cO^{(i-1)}\},
\]
with the indices $i$ taken modulo $n$.

\begin{lemma}
The dimension of $\cB(\cO)$ is $nr$ where $r$ is the rank of the Lie algebra $\g$.
\end{lemma}
\begin{proof}

Let $\hh^*_{\geq 0}$ be the positive Weyl chamber in the dual space to 
Cartan subalgebra of $\g$. For each generic orbit $\cO$ there is a unique 
representative $y\in \cO\cup \hh^*{>0}$. Let $x_1$ be such representative 
of $\cO^{(1)}$. Let us describe orbits $\cO^{(2)}$ such that $x_1+y_1\in \cO^{(2)}$
for some $y_1\in \cO_1$. Assume that $\cO^{(1)}$ is very large, i.e. $||x_1||>>1$.
Because $\cO_1$ is compact $||y_1||<C_1$ for some constant $C_1$ determined by the 
orbit $\cO_1$. Let $c^{(i)}_k$ be the value of $k$-th Casimir function on the orbit $\cO^{(i)}$.

For $k$-th Casimir function $c_k^{(2)}$ we have:
\[
c_k^{(2)}=c_k(x_1+y_1)=c_k^{(1)}+\sum_{i=1}^r \frac{\pa c_k(h)}{\pa h_i}|_{h=x_1}(y_1)_i+O(y^2)
\]
Because the matrix $\frac{\pa c_k(h)}{\pa h_i}$ is nondegerate for generic $h$, possible values
of the Euclidean vector with components $c_k^{(2)}$ span an $r$-dimensional 
neighborhood of $\{c_k^{(1)}\}$. 

Repeating this argument for each $\cO^{(i)}$ we conclude that each of $\cO_i$ is non-zero,
$\dim( \cB(\cO))=nr$. 

\end{proof}
%{\color{red} The formula $\dim(\cB(\cO))=nr$ is not true for all choices of coadjoint orbits $\cO=(\cO_1,\ldots,\cO_n)$. Here is a counterexample: write $\cB_n(\cO)$ for $\cB(\cO)$ then
%$\cB_n(\cO_1,\ldots,\cO_{n-1},\{0\})\simeq\cB_{n-1}(\cO_1,\ldots,\cO_{n-1})$ because $\cO_n=\{0\}$ and $\cO_n\subset\cO^{(n)}-\cO^{(n-1)}$ implies
%$\cO^{(n-1)}=\cO^{(n)}$. But this isomorphism is not compatible with the dimension formula. Do we need to impose conditions on $\cO$, if yes which? Or is superintegrability still ok, but with modified formulas for the dimensions of $\cP(\cO)$ and $\cB(\cO)$?}

Now let us describe the fiber $\cP(\cO; \cO^{(1)}, \dots, \cO^{(n)})$ of $p_{2,\cO}$ over $(\cO^{(1)}, \dots, \cO^{(n)})\in\cB(\cO)$:
\begin{equation*}
\begin{split}
\cP(\cO; \cO^{(1)},\dots, \cO^{(n)})&=\{(x,y)\in {\g^*}^{\times n}\times {\g^*}^{\times n}|x_i+y_i\in \cO_i, \ \ 
x_i\in \cO^{(i)}, y_i\in -\cO^{(i-1)} \}/G_n\\
&=\prod_{i=1}^n\bigl\{(x_i,y_i)\in\cO^{(i)}\times-\cO^{(i-1)}\,\, | \,\, x_i+y_i\in\cO_i\bigr\}/G
\end{split}
\end{equation*}
with the index $i$ taken modulo $n$ and with $G$ acting by the diagonal coadjoint action on $\cO^{(i)}\times-\cO^{(i-1)}$. 
Set 
\begin{equation}\label{cM3}
\cM(\cO^{(1)}, \cO^{(2)}, \cO^{(3)})=\{(x,y,z)\in \cO^{(1)}\times \cO^{(2)}\times \cO^{(3)}\,\,|\,\,x+y+z=0\}/G,
\end{equation}
with $G$ acting by the diagonal coadjoint action. Then 
 \begin{equation}\label{cMiso}
 \begin{split}
\bigl\{(x_i,y_i)\in\cO^{(i)}\times-\cO^{(i-1)}\,\, | \,\, x_i+y_i\in\cO_i\bigr\}/G&\overset{\sim}{\longrightarrow}\mathcal{M}(-\cO^{(i)},\cO^{(i-1)},\cO_i)\\
G(x_i,y_i)\mapsto G(-x_i,-y_i,x_i+y_i),
\end{split}
\end{equation}
%\[
%\cP(\cO; \cO^{(1)},\dots, \cO^{(n)})\simeq\cM(\cO_1, -\cO^{(1)}, \cO^{(2)})\times \cM(\cO_2, -\cO^{(2)}, \cO^{(3)})\times 
%\dots \times \cM(\cO_n, -\cO^{(n)}, \cO^{(1)}),
%\] 
and hence we conclude that
\begin{equation}\label{cpfactor}
\cP(\cO; \cO^{(1)},\dots, \cO^{(n)})\simeq\prod_{i=1}^n\mathcal{M}(-\cO^{(i)},\cO^{(i-1)},\cO_i),
\end{equation}
with the index $i$ taken modulo $n$. 
%%%%%%%%%%%%%%%%%%%%%%%%%%%%%%%%%%%
\begin{lemma}\label{dimrem}
Let $\cO^{(1)}, \cO^{(2)}$ be generic, sufficiently large coadjoint orbits and $\cO^{(3)}\neq 0$, then 
\[
\dim(\cM(\cO^{(1)}, \cO^{(2)}, \cO^{(3)}))=\dim(\cO^{(3)})-2r.
\]
\end{lemma}

\begin{proof} 
Let $x\in \cO^{(1)}$ be the unique representative which lies in the positive Weyl chamber.
Assume this orbit is "big", i.e. $||x||>>1$. The condition $x+y+z=0$ for $y\in \cO^{(2)}$ and 
$z\in \cO^{(3)}$ for a large orbit $\cO^{(2)}$ means that we have $r$ constraints 
$c_k(-x-y)=c_k^{(3)}$ on $y$. For large orbits $\cO^{(1)}$ and $\cO^{(2)}$ these 
constraints are independent. Taking into account that we are quotiening by $H$ we have
$\dim(\cM(\cO^{(1)},\cO^{(2)},\cO^{(3)})=\dim(\cO^{(3)}-2r$.
\end{proof}

\begin{corollary}
Thus the dimension of the fiber is $\dim(\cP(\cO; \cO^{(1)}, \dots, \cO^{(n)}))=\sum_{i=1}^n\dim(\cO_i)-2nr$.
\end{corollary}

%From here we have the balance of dimensions.

%%%%%%%%%%%%%%%%%%%%%%%%%%%%%%%%%
%\noindent
%{\color{red} The dimension formula in Remark \ref{dimrem} cannot be correct for all orbits $\cO^{(3)}$ (since the right hand side of the formula may be $<0$).
%For instance, if $\cO^{(3)}=\{0\}$ then $\dim(\cM(\cO^{(1)},\cO^{(2)},\{0\}))=0$.}

Each of the factors in \eqref{cpfactor} is the Hamiltonian reduction of the product of the three 
coadjoint orbits relative to the moment map of the diagonal coadjoint $G$-action, and therefore carries a natural symplectic structure. Moduli spaces $\cM(\cO^{(1)}, \cO^{(2)}, \cO^{(3)})$
and therefore fibers of $p_{2,\cO}$ are stratified symplectic spaces.

%Thus, we proved 
%%%%%%%%%%%%%%%%%%%%%%%%%%%%%%%%%%
\begin{theorem}\label{siperiodic} The Hamiltonian system generated by any Hamiltonian for the periodic spin CM chain described in section \ref{Ham-per} is superintegrable with the superintegrable structure described by the Poisson maps
\[
\cS(\cO) \stackrel{p_{1,\cO}}{\longrightarrow} \cP(\cO) \stackrel{p_{2,\cO}}{\longrightarrow} \cB(\cO)
\]
as introduced earlier in this section.
\end{theorem}
Here, as everywhere above, we assume that $\cO_i\neq \{0\}$ for each $i=1,\dots, n$.
%%%%%%%%%%%%%%%%%%%%%%%%%%%%%%%%%%
%\noindent
%{\color{red} Only state this for appropriate class of orbits $\cO=(\cO_1,\ldots,\cO_n)$, or modify the dimension formulas of $\cP(\cO)$ and $\cB(\cO)$ in the proof.}
\begin{proof} 
For $G_n(x,y)\in\cP(\cO)$
%\in\g^{*\times n}\times_{(\g^*/G)^{\times n}}\g^{*\times n}$
let $\widetilde{g}_i\in G$ such that $y_{i+1}=-Ad_{\widetilde{g}_i^{\,-1}}^*(x_i)$. Then 
\[
p_{1,\cO}^{-1}\bigl(G_n(x,y)\bigr)=\{G_n(x,g)\in\cS(\cO) \,\, | \,\, g_i\in\widetilde{g}_iZ_G(y_{i+1})\}
\]
(index $i$ taken modulo $n$)
which, generically, is isotropic and of dimension $nr=\dim(\cB(\cO))$. 

It remains to check the balance of dimensions. It follows from \eqref{radred} that 
\[
\dim(\cS(\cO))=\sum_{i=1}^n\dim(\cO_i).
\]
By Remark \ref{dimrem} we have, for generic $(\cO^{(1)},\ldots,\cO^{(n)})\in(\g^*/G)^{\times n}$,
\begin{equation*}
\begin{split}
\dim(\cP(\cO))&=\dim(\cB(\cO))+\dim(\cP(\cO;\cO^{(1)},\ldots,\cO^{(n)}))\\
&=\dim(\cB(\cO))+\sum_{i=1}^n\dim(\cO_i)-2nr.
\end{split}
\end{equation*}
Then
\[
\dim(\cP(\cO))+\dim(\cB(\cO))=\sum_{i=1}^n\dim(\cO_i)=\dim(\cS(\cO)),
\]
as desired.
%but this now immediately follows from
%the fact that $\dim(\cB(\cO))=nr$, Remark \ref{dimrem} and the fact that  (which follows, e.g., from 
%\eqref{radred}).
% and $\dim(\cB(\cO))=nr$.
%Furthermore, the generic fiber of $p_2$ is symplectic in view
%of \eqref{cpfactor}.
\end{proof}
%%%%%%%%%%%%%%%%%%%%%%

%%%%%%%%%%%%%%%%%%%%%%%
\begin{remark}\label{cyclicquantum} In the compact case, the quantum version of functions on $\cM(\cO^{(1)},\cO^{(2)},\cO^{(3)})$ is the algebra of endomorphisms $End((V_{\lambda_1}\otimes  V_{\lambda_2}\otimes V_{\lambda_3})^G)$ of the
subspace of $G$-invariant vectors in the tensor product $V_{\lambda_1}\otimes  V_{\lambda_2}\otimes V_{\lambda_3}$ with $V_{\lambda_i}$ the representation corresponding to $\cO^{(i)}$. 

The quantum version of the algebra of functions on the fiber $\cP(\cO; \cO^{(1)}, \dots, \cO^{(n)})$
is the algebra of endomorphisms of the vector space
\[
\textup{Hom}_G(V_{\lambda_1},V_{\lambda_n}\otimes V_{1})\otimes \textup{Hom}_G(V_{\lambda_2},V_{\lambda_1}\otimes V_{2})\otimes\cdots\otimes
\textup{Hom}_G(V_{\lambda_n},V_{\lambda_{n-1}}\otimes V_{n}).
%\otimes V_{\mu_n}, V_{\mu_1})\otimes Hom_\g(V_2\otimes V_{\mu_1}, V_{\mu_2})\otimes \dots \otimes Hom_\g(V_n\otimes V_{\mu_{n-1}}, V_{\mu_n}))
\]
%\[
%Hom_\g(V_1\otimes V_{\mu_n}, V_{\mu_1})\otimes Hom_\g(V_2\otimes V_{\mu_1}, V_{\mu_2})\otimes \dots \otimes Hom_\g(V_n\otimes V_{\mu_{n-1}}, V_{\mu_n}))
%\]
Here the  
%orbits $\cO^{(i)}$ correspond to $V_{\mu_i}$ and 
orbits $\cO_i$ correspond to $V_i$, and $\textup{Hom}_G(V_{\lambda_i},V_{\lambda_{i-1}}\otimes V_i)$ is the space of $G$-linear intertwiners
$V_{\lambda_i}\rightarrow V_{\lambda_{i-1}}\otimes V_i$.
For details see \cite{RSQCD1}.
\end{remark}
%%%%%%%%%%%%%%%%%%%%%

%%%%%%%%%%%%%%%%%%%%%%%%%%%%%
\subsection{Constructing solutions by the projection method and angle variables}\label{solutions}
%%%%%%%%%%%%%%%%%%%%%%%%%%%%%

For $\mathcal{H}$ a $G$-invariant function on $\g^*$, write $\mathcal{H}^{(i)}$ for the $G_n$-invariant function on $T^*(G^{\times n})$ defined
by $\mathcal{H}^{(i)}(x,g):=\mathcal{H}(x_i)$.
%It is easy to see that if $H$ is an invariant function on $\g^*$, 
The Hamiltonian flow through $(x,g)\in \g^{*\times n}\times G^{\times n}$ generated by $\mathcal{H}^{(i)}$ is
%evaluated on the $i$-th factor in $(\g^*\times G)^n$ is
\begin{equation}\label{evi-cl}
(x(t_i), g(t_i))=(x_1, \dots, x_n, g_1, \dots,g_{i-1}, e^{\nabla\mathcal{H}(x_i) t_i}g_i,g_{i+1} \dots, g_n)
\end{equation}
where $\nabla\mathcal{H}(x)\in \g$ is the gradient of $\mathcal{H}$ at $x\in\g^*$, i.e.,
% In other words, $\mathcal{H}(x)\in\g$ is the
%unique vector such that  
\[
y(\nabla\mathcal{H}(x))=\frac{d}{dt}\mathcal{H}(x+ty)\vert_{t=0}
\]
for all $y\in\mathfrak{g}^*$.
%$H\in C^G(x)$ 
%with respect to the Killing form $(\cdot,\cdot)$, i.e., 
%for $X\in \g$ we have
%\[
%(X, \nabla\mathcal{H}(x))=\frac{d}{dt} \mathcal{H}(x+tX)|_{t=0}.
%\]
The projection of such flow line to $T^*(G^{\times n})/G^{\times n}$ and further restricted to 
$\cS(\cO)\subset T^*(G^{\times n})/G_n$ is a flow line of the Hamiltonian vector field generated by 
the restriction of $\mathcal{H}^{(i)}$ to $\cS(\cO)$.

%\subsection{The construction of angle variables}\label{per-angles}
Now let us construct angle variables, i.e., functions on $S(\cO)$ which evolve linearly with respect to
the evolution (\ref{evi-cl}) for each $i=1,\ldots,n$. Write $\mathfrak{a}_+^*\subset\g^*$ for the elements $x\in\g^*$ which vanish on root spaces and satisfy $(x,\alpha)>0$
for $\alpha\in R_+$, where $(\cdot,\cdot)$ is the bilinear form on $\g^*$ induced by the Killing form. Write $\g^{\prime*}$ for the elements in $\g^*$ which are $G$-conjugate to an element in $\mathfrak{a}_+^*$, relative to the coadjoint action.

%positive Weyl chamber, i.e., $\mathfrak{a}_+$ consists of the elements $a\in\mathfrak{a}$ such that $\alpha(a)>0$ for all $\alpha\in R_+$. Then $\exp: \mathfrak{a}_+\overset{\sim}{\longrightarrow} A_+$. Write $\mathfrak{g}^\prime\subset\mathfrak{g}$ for the elements in $\mathfrak{g}$ which are $G$-conjugate to an element in $\mathfrak{a}_+$, relative to the adjoint action.
% \,\, | \,\, \alpha(a)>0 \quad \forall\,\alpha\in R_+\}$.
%Set $\mathfrak{b}:=\mathfrak{a}\oplus\bigoplus_{\alpha\in R_+}\mathfrak{g}_\alpha$.
%Choose a Borel subalgebra $\bf\subset \g$ and let $\hf\subset \bf\subset \g$ be the corresponding Cartan subalgebra. 

For $(x,g)\in \g^{\prime*\times n}\times G^{\times n}$ 
define $s_i\in G$ 
%$i=1,\dots, n$ 
by the property $Ad^*_{s_i}(x_i)\in\mathfrak{a}_+^*$. These elements are defined only up to $s_i\mapsto a_is_i$ where $a_i\in H$. Gauge transformations $h\in G_n$ act
%on $s=(s_1,\ldots,s_n)$ by 
by $(s_1,\ldots,s_n)\mapsto (s_1h_1^{-1},\ldots,s_nh_n^{-1})$.

Let $G_{\mathbb{C}}$ be a complexification of $G$, which we take to be connected.
Let $H_{\mathbb{C}}\subset G_{\mathbb{C}}$ be the Cartan subgroup $Z_{G_{\mathbb{C}}}(\mathfrak{h})$, where $\mathfrak{h}$ is the Cartan subalgebra
$\mathfrak{a}\oplus i\mathfrak{a}$ of the Lie algebra $\g_{\mathbb{C}}=\mathfrak{g}\oplus i\mathfrak{g}$ of $G_{\mathbb{C}}$. Then $A\subseteq H\subset H_{\mathbb{C}}$.
We identify $\mathfrak{g}^*$ with the real subspace of $\mathfrak{g}_{\mathbb{C}}^*$ consisting of the complex linear functionals that take real values on $\g$.
For finite dimensional $G_{\mathbb{C}}$-representations $V_1,\ldots,V_n$ choose vector $v_i\in V_i$ of $H_{\mathbb{C}}$-weight $\lambda_{i+1}$ and linear functionals $u_i^*\in V_i^*$
%,\ldots,v_n\in V_n$,
%$u^*_1\in V_1^*, \dots, u_n^**\in V_n^*$
of $H_{\mathbb{C}}$-weight $-\lambda_i$
%$\lambda_1,\dots, \lambda_n$ and $-\lambda_1, \dots, -\lambda_n$ respectively
(indices $i$ taken modulo $n$).\footnote{Recall that $G$ acts on dual vectors as $(gu^*)(v)=u^*(g^{-1}v)$.} Define
\begin{equation}\label{f-var}
f_{u,v}(x,g)=u_1^*(s_1g_1s_2^{-1}v_1) u_2^*(s_2g_2s_3^{-1}v_2) \dots u_n^*(s_ng_ns_1^{-1}v_n) 
\end{equation}
for $(x,g)\in (\g^{\prime*}\times G)^{\times n}$. This expression is well defined (i.e., invariant with respect to transformations
$s_i\to a_is_i$ with $a_i\in H$), and invariant with respect to gauge transformations. Thus, it defines a function on the subset
$(\g^{\prime*\times n}\times G^{\times n})/G_n$ of $T^*(G^{\times n})/G_n$. 

From the $G$-invariance of $\mathcal{H}$ we have the identity
% a finite dimensional representation $V$ and weight vectors $u_\mu^*\in V^*$ and $v_\lambda\in V$ 
%we have the identity
\[
u^*_i(s_ie^{t_i\nabla\mathcal{H}(x_i)}g_is_{i+1}^{-1}v_i)=e^{t_i\lambda_i(\nabla\mathcal{H}(y_i))}u^*_i(s_ig_is_{i+1}^{-1}v_i)
\]
where $y_i=Ad_{s_i}^*(x_i)\in\mathfrak{a}_+^*$, and consequently
%Immediately from the definition of these functions and the $G$-invariance of $\mathcal{H}$ 
%We conclude that 
\begin{equation}\label{towaa}
f_{u,v}(x(t_i),g(t_i))=e^{t_i\lambda_i(\nabla\mathcal{H}(y_i))}f_{u,v}(x,g).
\end{equation}
%where $y_i=Ad_{s_i}^*(x_i)$.
Logarithms of these functions evolve linearly, and hence they produce angle variables for the Hamiltonians $\mathcal{H}^{(i)}$ on $S(\cO)\cap (\g^{\prime*\times n}\times G^{\times n})/G_n$.
%$logarithms of such functions provide angle variables.
%an open dense subset of $(T^*G^n)/G^n$. 
%Denote by $f_{u,v,\cO}$ the restriction of $f_{u,v}$ to $S(\cO)\cap\bigl((\g^{\prime*}\times G)^{\times n}/G_n\bigr)$.

% Maximal independent set in the span of 
%logarithms of such functions gives angle variables. 

%%%%%%%%%%%%%%%%%%%%%%%%%%%%%%%%%%%%%%%%%
\section{Open spin Calogero-Moser chains}\label{open-CM}
%%%%%%%%%%%%%%%%%%%%%%%%%%%%%%%%%%%%%%%%%

Recall from the introduction that $G$ is a split real connected Lie group with finite center which admits a complexification,
and $K\subset G$ is the subgroup of fixed points of a fixed Cartan involution $\Theta$ of $G$. Recall furthermore the root space decomposition 
$\g=\af\oplus \bigoplus_{\alpha>0}(\RR e_\alpha\oplus \RR e_{-\alpha})$ with the Cartan subalgebra $\af\subset\g$ and the root vectors $e_\alpha\in\g_\alpha$ 
% is the root decomposition
%of the Lie algebra $\g=Lie(G)$, 
such that the infinitesimal Cartan involution $\theta$ acts as
\[
\theta(h)=-h, \ \ \theta(e_\alpha)=-e_{-\alpha}
\]
for $h\in\af$ and $\alpha\in R$. We will furthermore normalise the root vectors in such a way that $(e_\alpha,e_{-\alpha})=1$, with $(\cdot,\cdot)$ the Killing form of $\g$.

To avoid cumbersome notations, we will not always indicate in notations that we are in the open case. This leads to an overlap of some of the notations with the ones for the periodic case. For instance, the moment maps, Poisson spaces and Poisson projections will be denoted in the same way. 
%We also assume that we are in a {\bf completely split case}\footnote{ {\bf Remind to non-experts what is means}.}.

%%%%%%%%%%%%%%%%%%%%%%%%%%%%%%%%%%%%%%%%%%%%
\subsection{The phase space as the Hamiltonian reduction} \label{ph-space-open}
%%%%%%%%%%%%%%%%%%%%%%%%%%%%%%%%%%%%%%%%%%
Consider for $n\geq 0$ the manifold $T^*(G^{\times n+1})$ with the standard symplectic structure.
We trivialize the cotangent bundle $T^*(G^{\times n+1})$ by right translations:
\begin{equation}\label{trivialisationK}
T^*(G^{\times n+1})\simeq (T^*G)^{\times n+1}\simeq {\g^*}^{\times n+1}\times G^{\times n+1}
\end{equation}

We have a natural action of $K\times G^{\times n}$ on $G^{\times n+1}$ by left translations:
\[
(k_\ell, h_1, \dots, h_n)_L(g_0, g_1, \dots, g_n)=(k_\ell g_0, h_1g_1,\dots, h_ng_n)
\]
This action lifts to the following Hamiltonian action on $T^*G^{\times n+1}$,
\begin{equation*}
\begin{split}
(k_\ell, h_1, \dots, h_n)_L(x_0, &\dots, x_n,g_0, g_1, \dots, g_n)=\\
&= (Ad^*_{k_\ell}(x_0), Ad^*_{h_1}(x_1), \dots, Ad^*_{h_n}(x_n),k_\ell g_0, h_1g_1,\dots, h_ng_n)
\end{split}
\end{equation*}
with the moment map
\[
\mu_L(x,g)=(\pi(x_0), x_1, \dots, x_n)
\]
where the projection $\pi: \g^*\rightarrow \kf^*$  is the dual map dual to the embedding $\kf\hookrightarrow \g$.

Similarly, the action of $G^{\times n}\times K$ on $G^{\times n+1}$ by right translations
\[
(h_1,\dots, h_n, k_r)_R(g_0,g_1,\dots, g_n)=(g_0h_1^{-1}, g_1h_2^{-1},\dots, g_{n-1}h_n^{-1}, g_nk_r^{-1})
\]
lifts to the following Hamiltonian action on $T^*G^{\times n+1}$,
\begin{equation*}
\begin{split}
(h_1,\dots, h_n, k_r)_R(x_0,& \dots, x_n, g_0,g_1,\dots, g_n)\\
&=(x_0, x_1,\dots, x_n,g_0h_1^{-1}, g_1h_2^{-1},\dots, g_{n-1}h_n^{-1}, g_nk_r^{-1}),
\end{split}
\end{equation*}
with the moment map
\[
\mu_R(x,g)=(-Ad^*_{g_0^{-1}}(x_0), -Ad^*_{g_1^{-1}}(x_1),\dots, -Ad^*_{g_{n-1}^{-1}}(x_{n-1}),-\pi(Ad^*_{g_n^{-1}}(x_n))).
\]

As a result, the group $G_{n,K}:=K\times G^{\times n}\times K$
acts on $T^*(G^{\times n+1})$ as 
\begin{equation}\label{k-n-act}
\begin{split}
(k_\ell,h_1,\dots,& h_{n}, k_r)(x_0,\dots, x_n, g_0,\dots, g_n)=\\ 
=&(Ad^*_{k_\ell}(x_0), Ad^*_{h_1}(x_1),\dots, Ad^*_{h_{n}}(x_n), k_\ell g_0h_1^{-1}, h_1g_1h_2^{-1}, \dots, h_{n}g_nk_r^{-1})
\end{split}
\end{equation}
with $k_\ell,k_r\in K$ and $h_i\in G$. This action is Hamiltonian with the moment map $\mu: T^*(G^{\times n})\to \kf^* \times {\g^*}^{\times n}\times \kf^*$ given by
\begin{equation}\label{momentK}
\begin{split}
\mu((&x_0,\dots, x_n,g_0,\dots, g_n)=(\mu_L(x,g),0)+ (0,\mu_R(x,g))= \\ 
&=(\pi(x_0), x_1-Ad^*_{g_0^{-1}}(x_0), x_2-Ad^*_{g_1^{-1}}(x_1),\dots, x_{n}-Ad^*_{g_{n-1}^{-1}}(x_{n-1}), -\pi(Ad^*_{g_n^{-1}}(x_n))).
\end{split}
\end{equation}
For $n=0$ this is the $K\times K$ action $(k_\ell, k_r)(x,g)=(Ad^*_{k_\ell}(x), k_\ell g k_r^{-1})$ on $T^*G$,
with the moment map $(x,g)\mapsto (\pi(x), -\pi(Ad_{g^{-1}}(x)))$.
It is easy to check explicitly that this moment map intertwines
%is invariant with respect to the 
%action of $G_{n,K}$, intertwining
the action of $G_{n,K}$ on $T^*(G^{\times n+1})$ given by (\ref{k-n-act}) with its diagonal coadjoint action on 
%$\g^*_{n,K}\simeq 
$\kf^* \times {\g^*}^{\times (n-1)}\times \kf^*$.

Because the action of $G_{n,K}$ on $T^*(G^{\times n+1})$ is Hamiltonian, the space 
$T^*(G^{\times n+1})/G_{n,K}$\footnote{Recall that here and in everywhere else in this paper $X/H$ means the
GIT quotient for a Lie group $H$ action on a manifold $X$.} is Poisson with symplectic leaves being given by the Hamiltonian reduction with respect to the moment
map \eqref{momentK}. Let $\cO=(\cO_\ell^K\times  \cO_1 \times \dots, \cO_{n} \times \cO^K_r)$ with $\cO_i\subset \g^*$ coadjoint $G$-orbits and 
$\cO_\ell^K, \cO_r^K \subset \kf^*$ coadjoint $K$-orbits, then the corresponding symplectic leaf in $T^*(G^{\times n+1})/G_{n,K}$ is
\begin{equation}\label{bcase}
\begin{split}
\cS&(\cO)=\mu^{-1}(\cO)/G_{n,K}\\
&=\bigl\{(x_0,\dots, x_n, g_0,\dots, g_n) \in {\g^*}^{\times n+1}  \times G^{\times n+1}\,\, |\,\, \pi(x_0)\in \cO^K_\ell, \, -\pi(Ad_{g_n^{-1}}^*(x_n))\in \cO^K_r,\\
&\qquad\qquad\qquad\qquad\qquad\quad x_1-Ad^*_{g_0^{-1}}(x_0) \in \cO_1,\dots,  x_{n}-Ad^*_{g_{n-1}^{-1}}(x_{n-1})\in \cO_n\bigr\}/G_{n, K}.
\end{split}
\end{equation}
Each symplectic leaf $\S(\cO)$ is a stratified symplectic space and is the phase space for the corresponding open spin Calogero-Moser chain. We will describe the largest stratum of $\cS(\cO)$ later.

%%%%%%%%%%%%%%%%%%%%%%%%%%%%%%%%%%%%%%%%%%%%%%%%%
\subsection{The Hamiltonians of the open spin Calogero-Moser chain}\label{oCM-Ham}
%%%%%%%%%%%%%%%%%%%%%%%%%%%%%%%%%%%%%%%%%%%%%%%%%
After the trivialization \eqref{trivialisationK} of $T^*(G^{\times n+1})$ by right translations we have a natural Poisson projection $ T^*(G^{\times n+1})\to {\g^*}^{\times n+1}$ to the first factor. It is $G_{n, K}$-invariant with the following action of $G_{n,K}$ on ${\g^*}^{\times n+1}$
\[
(k_\ell, h_1,\dots, h_{n}, k_r): (x_0, x_1,\dots, x_n)\mapsto (Ad^*_{k_\ell}(x_0), Ad^*_{h_1}(x_1),\dots, Ad^*_{h_{n}}(x_n)).
\]
%This action on the right is not symmetric with respect to $k_\ell, k_r$ because the gauge action is not symmetric.
This gives rise to the projection 
\[
p: T^*(G^{\times n+1})/G_{n,K}\to (\g^*/G)^{\times n+1}
\]
which is Poisson because it is the composition of natural Poisson projections 
\begin{equation}\label{pi-project}
 T^*(G^{\times n+1})/G_{n,K}\to ({\g^*}^{\times n+1})/G_{n,K}=\g^*/K\times (\g^*/G)^{\times n} \to (\g^*/G)^{\times n+1}.
\end{equation}
Here the Poisson structure on the right is  trivial (the Poisson tensor is zero). The last projection is a consequence of the embedding $K\hookrightarrow G$.

Restricting this projection to the symplectic leaf $\cS(\cO)$ we have the Poisson projection
\begin{equation}\label{p}
p_\cO: \cS(\cO)\to \cB(\cO), \quad G_{n,K}(x_0,\dots,x_n, g_0,\dots, g_n)\mapsto (Gx_0,\dots, Gx_n)
\end{equation}
where $\cB(\cO)\subset (\g^*/G)^{\times n+1}$ is, by definition, the image of $\cS(\cO)$.
It can be described explicitly from the description of $\cS(\cO)$ as
\begin{equation}\label{cBO}
\begin{split}
\cB(\cO)=\bigl\{(\cO^{(0)},\dots, \cO^{(n)})\in (\g^*/G)^{\times n+1}&\,\, |\,\, \cO_\ell^K\subseteq \pi(\cO^{(0)}),\cO_r^K\subseteq -\pi(\cO^{(n)}),\\ 
&\cO_1\subseteq\cO^{(1)}-\cO^{(0)},\ldots,\cO_n\subseteq\cO^{(n)}-\cO^{(n-1)}\bigr\}.
%\cB(\cO)=\{(\cO^{1)}, \dots, \cO^{n+1)})|\cO_\ell^K\subset \pi(\cO^{(1)}), \cO^{(2)}-\cO^{(1)}\subset \cO_1,
%\dots, \\ \nonumber \cO^{(k+1)}-\cO^{(k)}\subset \cO_k, \dots, \cO^{(n+1)}-\cO^{(n)}\subset \cO_n,  \cO_r^K\subset -\pi(\cO^{(n+1)})
\end{split}
\end{equation}
%where $\cO^{(k)}=Gx_{k}$.

Hamiltonians of the open spin Calogero-Moser system are pull back $p^*$ of functions on ${(\g^*/G)}^{\times n+1}$ restricted to $\cS(\cO)$. 

The subalgebra of Hamiltonians is a Poisson commuting subalgebra. Quadratic Hamiltonians are given by Casimir functions. We will compute the radial components of the quadratic Casimirs explicitly in section \ref{qu-Ham}. Hamiltonians are constant on fibers of the projection $p_{\cO}$. 

%%%%%%%%%%%%%%%%%%%%%%%%%%%%%%
\subsection{The gauge fixing}\label{op-g-reduct}
%%%%%%%%%%%%%%%%%%%%%%%%%%%%%%%
Fix $i\in\{0,\ldots,n\}$. For $g=(g_0, \dots, g_n)\in G^{\times n+1}$ and $k_\ell,k_r\in K$ define $h\in G^{\times n}$ (depending on $i,g,k_\ell,k_r$) by
\begin{equation*}
h_j=
\begin{cases}
k_\ell g_0g_1\cdots g_{j-1}\qquad &\hbox{ if }\,\, j\leq i,\\
k_r g_n^{-1}g_{n-1}^{-1}\cdots g_{j}^{-1}\qquad &\hbox{ if }\,\, j>i.
\end{cases}
\end{equation*}
Then $(k_\ell, h_1, \dots, h_n, k_r)$ acts on $g\in G^{\times n+1}$ as
\[
g\mapsto (1,\dots, 1, k_\ell g_0g_1\dots g_n k_r^{-1}, 1,\dots, 1),
\]
Here the nontrivial entry is at the position $i$.
This gives an $i$-independent isomorphism
\begin{equation}\label{g-isom}
G^{\times n+1}/G_{n,K}\to K\backslash G/K,  \qquad G_{n,K}(g_0,\dots, g_n)\mapsto Kg_0g_1\dots g_nK.
\end{equation}
For the cotangent bundles the gauge fixing gives an isomorphism
%From now on we will work with representatives of $G_{n,K}$-orbits in $G^{\times n+1}/G_{n,K}$ with all coordinates $1$, except for the last coordinate. 
%Lifting $\phi$ to the cotangent bundles gives an
%isomorphism
\begin{equation}\label{cot-map}
\phi_i: \bigl(\g^{*\times n+1}\times G^{\times n+1}\bigr)/G_{n,K}\overset{\sim}{\rightarrow} K\backslash ( {\g^*}^{\times {n+1}}\times G)/K
\end{equation}
mapping the $G_{n,K}$-orbit through $(x,g)\in \g^{*\times n+1}\times G^{\times n+1}$ to the double $K$-coset through
\[
\bigl(x_0,Ad_{g_0}^*(x_1),\ldots,Ad_{g_0\cdots g_{i-1}}^*(x_i),Ad_{g_n^{-1}\cdots g_{i+1}^{-1}}^*(x_{i+1}),\ldots,Ad_{g_n^{-1}}^*(x_n),g_0g_1\cdots g_n\bigr)
\]
For example for $i=n$ this expression is $\bigl(x_0,Ad_{g_0}^*(x_1),\ldots,Ad_{g_0\cdots g_{n-1}}^*(x_n),g_0g_1\cdots g_n\bigr)$. In (\ref{cot-map}) the double $K$-cosets in the codomain of 
$\phi_i$ are taken relative to the 
$i$-dependent $K\times K$-action 
\[
(k_\ell,k_r)(x_0,\ldots,x_n,g)=(Ad_{k_\ell}^*(x_0),\cdots,Ad_{k_\ell}^*(x_i),Ad_{k_r}^*(x_{i+1}),\ldots,Ad_{k_r}^*(x_n),k_\ell g k_r^{-1})
\]
on ${\g^*}^{\times n+1}\times G$. 
%After trivialisation of $T^*(G^{\times n+1})$ by right translations (see \eqref{trivialisationK})
%the isomorphism $\phi$ acts by
%\begin{eqnarray}
%\begin{equation*}
%G_{n,K} (x_0,x_1,\dots, x_n, g_0, g_1,\dots, g_n)\mapsto 
%K(x_0, Ad^*_{g_0}(x_1), Ad^*_{g_0g_1}(x_2), \\ \nonumber \dots, Ad^*_{g_0g_1\dots g_{i-1}}(x_i), Ad^*_{g_n^{-1}g_{n-1}^{-1}\dots g_{i+1}^{-1}}(x_{i+1}),\dots, Ad^*_{g_n^{-1}}(x_n), g_0g_1\dots g_n)K
%G_{n,K} (x_0,x_1,\dots, x_n, g_0, g_1,\dots, g_n)\mapsto 
%K(x_0, Ad^*_{g_0}(x_1), Ad^*_{g_0g_1}(x_2), \\ \nonumber \dots, Ad^*_{g_0g_1\dots g_{i-1}}(x_i), Ad^*_{g_n^{-1}g_{n-1}^{-1}\dots g_{i+1}^{-1}}(x_{i+1}),\dots, Ad^*_{g_n^{-1}}(x_n), g_0g_1\dots g_n)K
%\end{equation*}
%Here $K\times K$ acts on ${\g^*}^{\times {n+1}}\times G$ as 
%\[
%(k_\ell, k_r): (x_0, x_1, \dots, x_n; g) \to (Ad^*_{k_\ell}, \dots, Ad^*_{k_\ell}(x_n),; k_\ell gk_r^{-1})
%\]

Now we can describe the symplectic leaf $\cS(\cO)$ as a subvariety in $K\backslash ( {\g^*}^{\times {n+1}}\times G)/K$ though the isomorphism $\varphi_n$ as
\begin{equation}\label{ps}
\begin{split}
\cS(\cO)=K\backslash\bigl\{(y_0,y_1,\ldots,y_n,g)\in\g^{*\times n+1}\times G& \,\, | \,\, \pi(y_0)\in\cO_\ell^K, -\pi\bigl(Ad_{g^{-1}}^*(y_n))\in\cO_r^K,\\
&\quad\,\,\, y_1-y_0\in\cO_1,\ldots,y_n-y_{n-1}\in\cO_n\bigr\}/K.
\end{split}
\end{equation}
Note that, as in the periodic case, $\cS(\cO)$ is a symplectic stratified space. From now on we will 
focus mostly on the largest stratum $\cS(\cO)_{reg}$.
%\begin{eqnarray*}
%\cS(\cO)=K\backslash\{(x^{(0)},x^{(1)} \dots, x^{(n)}, g)\in \g^*\times \dots \times \g^*\times G| \\
%\pi(x^{(0)})\in \cO^K_1, x^{(1)}-x^{(0)}\in \cO_1, \dots,  x^{(k+1)}-x^{(k)}\in \cO_k, \\ \dots, x^{(n)}-x^{(n-1)}\in \cO_k,-\pi(Ad_{g^{-1}}^*(x^{(n)}))\in \cO^K_2\}/K
%\end{eqnarray*}

%%%%%%%%%%%%%%%%%%%%%%%%%%
\subsection{The regular part of the symplectic leaf $\cS(\cO)$} \label{op-s-leaf}
%%%%%%%%%%%%%%%%%%%%%%%%%%
We use the gauge fixing isomorphism $\varphi_n$ in the remainder of the text.
%From now on we assume that $K=G^\Theta$ where $\Theta$ is the Cartan involution on $G$ and we also assume that we are in the completely split case.
We now use $K\backslash G/K\simeq A/W$ with $W=N_K(A)/M$ the Weyl group of $G$ (see subsection \S3 of the introduction)
to describe the regular part of the symplectic leaf $\cS(\cO)$
in radial coordinates.

Define the regular part of the phase space $\cS(\cO)$ (see \eqref{ps}) as 
\[
\cS(\cO)_{reg}=\cS(\cO)\cap K\backslash (\g^{*\times n+1}\times G_{reg})/K.
\]
The regular part $\cS(\cO)_{reg}\subset \cS(\cO)$ is its largest stratum of the
stratified symplectic space $\cS(\cO)$.

We can then choose a representative of $K(y_0,\ldots,y_n,g)K\in\cS(\cO)_{reg}$ with the $G$-component in $A_{reg}$ by writing $g=k_\ell ak_r^{-1}$
and $y_i=Ad_{k_\ell}^*(x^{(i)})$ with $k_\ell,k_r\in K$ and $a\in A_{reg}$. It follows that
%, $N_K(A)\subset K$ is the normalizer and $Z_K(A)\subset K$ is the centralizer of $A$ in $K\subset G$.
%First note that
\begin{equation*}
\begin{split}
\cS(\cO)_{reg}\simeq\bigl\{(x^{(0)}, \dots, x^{(n)}, a)\in \g^{*\times n+1}\times &A_{reg}\,\, | \,\,
\pi(x^{(0)})\in \cO^K_\ell,-\pi(Ad_{a^{-1}}^*(x^{(n)}))\in \cO^K_r,\\
&\,\, x^{(1)}-x^{(0)}\in\cO_1,\ldots,x^{(n)}-x^{(n-1)}\in\cO_n\bigr\}/N_K(A).
\end{split}
\end{equation*}
Here $N_K(A)$ acts diagonally on $\g^{*\times n+1}\times A_{reg}$ via the coadjoint action on $\g^*$ and the conjugation action on $A_{reg}$.\footnote{We use here the fact that $k_\ell a k_r^{-1}=k_\ell^\prime a^\prime k_r^{\prime-1}$ for $k_\ell,k_\ell^\prime,k_r,k_r^\prime\in K$ and $a,a^\prime\in A_{reg}$ is implying that $k_\ell^{-1}k_\ell^\prime=k_r^{-1}k_r^\prime\in N_K(A)$,  cf., e.g., \cite[\S VII.3]{Kn}. This essentially follows from the global Cartan decomposition of $G$.} We can now also divide out the action of $M=Z_K(A)$,
to obtain the isomorphism
\begin{equation*}
\begin{split}
\cS(\cO)_{reg}\simeq\bigl\{(M(x^{(0)}, \dots, x^{(n)}), a)\in \g^{*\times n+1}/M\times &A_{reg}\,\, | \,\,
\pi(x^{(0)})\in \cO^K_\ell,-\pi(Ad_{a^{-1}}^*(x^{(n)}))\in \cO^K_r,\\
&\,\, x^{(1)}-x^{(0)}\in\cO_1,\ldots,x^{(n)}-x^{(n-1)}\in\cO_n\bigr\}/W,
\end{split}
\end{equation*}
where $M$ acts by the diagonal coadjoint action on $\g^{*\times n+1}$ and $W$ acts diagonally on the space $\g^{*\times n+1}/M\times A_{reg}$ in the natural way.

%The regular part of the symplectic leaf $\cS(\cO)$ is the intersection $\cS(\cO)_{reg}=\cS(\cO)\cap 
%K\backslash ( {\g^*}^{\times {n}}\times G_{reg})/K$. Using Cartan decomposition we have $K(\mathfrak{g}^{*\times n}\times G_{reg})/K\simeq ((\mathfrak{g}^{*\times n})\times A_{reg})/N_K(A)$
%\footnote{The proof uses some particular facts on the uniqueness of the $A$-and $K$-components in the Cartan decomposition $g=k_1ak_2$ of a group element $g\in G (k_1,k_2\in K; a\in A)$; these are the properties given in Theorem 7.39 of Knapp's book "Lie groups beyond an introduction".} 
%\begin{align*}
%\cS(\cO)_{reg}&=\{(x^{(1)}, \dots, x^{(n)}, )\in \g^*\times \dots \times \g^*\times A_{reg}| \\
%&p(x^{(1)})\in \cO^K_1, \dots,  x^{(k+1)}-x^{(k)}\in \cO_k, \dots, -p(Ad_{a^{-1}}^*(x^{(n)}))\in \cO^K_2\}/N_K(A)
%\end{align*}

%Choose a basis $e_\alpha$ in each root subspace $\g_\alpha\subset \g$. For $x\in \g^*$ 
%Recall that we write $x_\alpha=x(e_\alpha)$ and $x_0=x|_{\af}\in \af^*$. For $y\in \kf^*$ we write $y_\alpha=y(e_\alpha-e_{-\alpha})$ with $y_\alpha=-y_{-\alpha}$.
Recall that we identified $\mathfrak{g}\simeq\mathfrak{g}^*$ and $\mathfrak{a}\simeq\mathfrak{a}^*$ via the Killing form, so that $x\in\g^*$ corresponds to
$x_0+\sum_\alpha x_\alpha e_\alpha$ with $x_0$ the element in $\mathfrak{a}$ corresponding to $x\vert_{\mathfrak{a}}\in\mathfrak{a}^*$ and $x_\alpha=x(e_{-\alpha})$. 
Denote by $x_0^{(k)}, x_\alpha^{(k)}$ the components of vectors $x\in\g^{*\times n+1}$ from the $k$-th factor of ${\g^*}^{n+1}$, and
$\mu_0^{(k)}, \mu_\alpha^{(k)}$ the components of $\mu\in \cO_k$.
% write $x_\alpha=x(e_\alpha)$ and $x_0=x|_{\af}\in \af^*$. 
For $y\in \kf^*$ we write $y_{[\alpha]}=y(e_{-\alpha}-e_\alpha)$, so that $y_{[\alpha]}=-y_{[-\alpha]}$.
% and by $\mu_{[\alpha]}', \mu_{[\alpha]}''$  the components of $\mu'\in \cO_\ell^K$ and $\mu''\in -\cO_r^K$ in the basis $e_\alpha-e_{-\alpha}$. 

Consider 
\begin{equation*}
\begin{split}
T(\cO)_{reg}=\bigl\{(x^{(0)}, \dots, x^{(n)}, a)\in \g^{*\times n+1}\times A_{reg} \,\,|\,\,& \pi(x^{(0)})\in \cO^K_\ell,-\pi(Ad_{a^{-1}}^*(x^{(n)}))\in \cO^K_r,\\
&x^{(1)}-x^{(0)}\in\cO_1,\ldots, x^{(n)}-x^{(n-1)}\in\cO_n\bigr\}.
% \dots, \\ \nonumber x^{(k+1)}-x^{(k)}\in \cO_k, \dots, \}
%\end{eqnarray}
\end{split}
\end{equation*}
%Similarly define $T(\cO)_{reg}$. 
Clearly 
%$S(\cO)=T(\cO)/N_K(A)$ and 
$S(\cO)_{reg}=T(\cO)_{reg}/N_K(A)$.

For $(x^{(0)},\ldots,x^{(n)},a)\in T(\cO)_{reg}$ write $\mu^\prime=\pi(x^{(0)})\in\cO_\ell^K$, $\mu^{\prime\prime}=-\pi(Ad_{a^{-1}}^*(x^{(n)}))\in\cO_r^K$ and
$\mu^{(i)}=x^{(i)}-x^{(i-1)}\in\cO_i$ for $i=1,\ldots,n$.
%In the space of functions on $T(\cO)_{reg}$ for 
The Cartan components of $x^{(k)}$ and their root coordinates then satisfy
%we have for $\mu^\prime\in\cO_\ell^K$, $\mu^{\prime\prime}\in\cO_r^K$
%and $\mu_i\in\cO_i$ for $i=1,\ldots,n$,
\begin{equation}\label{x-mu}
\begin{split}
x_\alpha^{(0)}-x_{-\alpha}^{(0)}&=\mu_{[\alpha]}^\prime,\qquad\quad a_\alpha x_{-\alpha}^{(n)}-a_\alpha^{-1}x_\alpha^{(n)}=\mu_{[\alpha]}^{\prime\prime},\\
x^{(i)}_\alpha-x^{(i-1)}_\alpha&=\mu_\alpha^{(i)},\qquad\quad x_0^{(i)}-x_0^{(i-1)}=\mu_0^{(i)}
\end{split}
\end{equation}
for $i=1,\ldots,n$.
%\begin{align}\label{x-mu}
%x^{(1)}_\alpha-x^{(1)}_{-\alpha}&=\mu'_\alpha,  \ \
%x^{(2)}_\alpha-x^{(1)}_\alpha=\mu^{(1)}_\alpha, \ \ x^{(2)}_{-\alpha}-x^{(1)}_{-\alpha}=\mu^{(1)}_{-\alpha},  
%\dots, \\ \nonumber
%x^{(n)}_\alpha-x^{(n-1)}_\alpha&=\mu^{(n-1)}_\alpha, \ \ x^{(n)}_{-\alpha}-x^{(n-1)}_{-\alpha}=\mu^{(n-1)}_{-\alpha}, \ \
%a_\alpha x^{(n)}_\alpha-a_{-\alpha}x^{(n)}_{-\alpha}=\mu''_\alpha. 
%\end{align}
%\begin{align*}
%x^{(2)}_0-x^{(1)}_0=\mu^{(1)}_0,  \ \
%x^{(3)}_0-x^{(2)}_0=\mu^{(2)}_0,  \ \
%\dots, \ \
%x^{(n)}_0-x^{(n-1)}_0=\mu^{(n-1)}_0
%\end{align*}

%Here $\mu^{(i)}\in \cO_i, \mu'\in \cO^K_1, \mu''\in \cO^K_2$,  $\mu^{(i)}_\alpha$, $\mu^{(i)}_0$ are 
%coordinate functions on $\cO_i\subset \g^*$ is the root basis, $\mu^{(i)}_)$ and
%$\mu_\alpha'$ and $\mu_\alpha''$ are coordinates on $\cO^K_1\subset \kf^*$ and $\cO^K_1\subset \kf^*$
%respectively corresponding to positive roots of $\g$ (see the Introduction for details). 

%Now let us use these equations to express coordinates $x^{(i)}_{\pm \alpha}$ and Cartan components $x^{(i)}_0$ in terms of
%coordinates on orbits $\cO^K_1, \cO_i, \cO^K_2$ .
%In particular, we have:
It is easy to solve the equations for Cartan parts $x_0^{(i)}$ ($0<i<n$) 
%and express $x_0^{(k)}$
in terms of Catran components of $x^{(0)},x^{(n)}$ and $\mu^{(j)}$, 
\begin{equation}\label{x-var}
x^{(i)}_0=\mu^{(i)}_0+\dots + \mu^{(1)}_0+ x^{(0)}_0=x^{(n)}_0-\mu^{(n)}_0-\dots-\mu^{(i+1)}_0
\end{equation}
%\begin{equation}\label{x-var}
%x^{(i)}_0=\mu^{(i-1)}_0+\dots + \mu^{(1)}_0+ x^{(1)}_0=x^{(n)}_0-\mu^{(i)}_0-\dots-\mu^{(n)}_0
%\end{equation}
%Let $C(T(\cO))$ be the algebra of functions which are polynomials in $x^{(k)}$ and rational 
%in $A_{reg}$.

%%%%%%%%%%%%%%%%%%%%%%%%%%%%
\begin{proposition} The following identities hold for $\alpha\in R$ and $k=0,1,\ldots,n$:
% in $C(T(\cO))$:
\begin{equation}\label{x-k-b}
x^{(k)}_\alpha=K_\alpha+\sum_{l=1}^{k}\frac{a_{\alpha}\mu^{(l)}_\alpha-a_{\alpha}\mu^{(l)}_{-\alpha}}{a_\alpha-a_{\alpha}^{-1}}
%-\frac{a_{\alpha}\mu^{(k)}_\alpha-a_{-\alpha}\mu^{(k)}_{-\alpha}}{a_\alpha-a_{-\alpha}}
+\sum_{l=k+1}^n\frac{a_{\alpha}^{-1}\mu^{(l)}_\alpha-a_{\alpha}\mu^{(l)}_{-\alpha}}{a_\alpha-a_{\alpha}^{-1}}
\end{equation}
where 
\begin{equation}\label{Ka}
K_\alpha=\frac{a_\alpha\mu_{[\alpha]}^\prime+\mu''_{[\alpha]}}{a_\alpha-a_{\alpha}^{-1}}.
\end{equation}
%\begin{equation}\label{x-k-b}
%x^{(k)}_\alpha=K_\alpha-\sum_{l<k}\frac{a_{-\alpha}\mu^{(l)}_\alpha-a_{-\alpha}\mu^{(l)}_{-\alpha}}{a_\alpha-a_{-\alpha}}
%-\frac{a_{\alpha}\mu^{(k)}_\alpha-a_{-\alpha}\mu^{(k)}_{-\alpha}}{a_\alpha-a_{-\alpha}}
%-\sum_{l>k}\frac{a_{\alpha}\mu^{(l)}_\alpha-a_{-\alpha}\mu^{(l)}_{-\alpha}}{a_\alpha-a_{-\alpha}}
%\end{equation}
%where 
%\begin{equation}\label{Ka}
%K_\alpha=\frac{\mu''_\alpha-a_{-\alpha}\mu'_\alpha}{a_\alpha-a_{-\alpha}}.
%\end{equation}
\end{proposition}
%%%%%%%%%%%%%%%%%%%%%%%%%%%%%%
\begin{proof}
Denote 
\begin{equation}\label{muexp}
\mu=\mu^{(1)}+\dots + \mu^{(n)}.
\end{equation}
Note that $x^{(n)}-x^{(0)}=\mu$.
%From (\ref{x-var}) we immediately conclude that 
%\[
%x^{(n)}-x^{(1)}=\mu
%\]

Fix $\beta\in R_+$. For $x^{(1)}_{\pm\beta}$ and $x^{(n)}_{\pm\beta}$ the formula $x^{(n)}-x^{(0)}=\mu$ implies
\[
x^{(n)}_\beta-x^{(0)}_\beta=\mu_\beta, \ \ x^{(n)}_{-\beta}-x^{(0)}_{-\beta}=\mu_{-\beta}.
\]
Combined with the first line of \eqref{x-mu} we end up with four linear equations in $x^{(0)}_\beta,x^{(0)}_{-\beta},x^{(n)}_\beta,x^{(n)}_{-\beta}$
which, by the assumption that $a$ is regular, are uniquely solved by 
\begin{equation}\label{extremes}
\begin{split}
x^{(0)}_\alpha&=\frac{a_\alpha\mu_{[\alpha]}^\prime+\mu_{[\alpha]}''+(a_{\alpha}^{-1}\mu_\alpha-a_{\alpha}\mu_{-\alpha})}{a_\alpha-a_{\alpha}^{-1}},\\
%x^{(0)}_{-\alpha}&=\\
x^{(n)}_{\alpha}&=\frac{a_\alpha\mu_{[\alpha]}^\prime+\mu_{[\alpha]}''+(a_{\alpha}\mu_\alpha-a_\alpha\mu_{-\alpha})}{a_\alpha-a_{\alpha}^{-1}}\\
%x^{(n)}_{-\alpha}
\end{split}
\end{equation}
for $\alpha=\beta,-\beta$ (here we used that $a_{-\beta}=a_\beta^{-1}$, $\mu_{[-\beta]}^\prime=-\mu_{[\beta]}^\prime$ and 
$\mu^{\prime\prime}_{[-\beta]}=-\mu^{\prime\prime}_{[\beta]}$).
%From equations (\ref{x-mu}) we have:
%\[
%x^{(1)}_{\alpha}-x^{(1)}_{-\alpha}=\mu_{\alpha}',\qquad
%a_\alpha x^{(n)}_\alpha-a_{-\alpha}x^{(n)}_{-\alpha}=\mu_\alpha''.
%\]
%\[
%a_\alpha x^{(n)}_\alpha-a_{-\alpha}x^{(n)}_{-\alpha}=\mu_\alpha'', \ \ x^{(1)}_{\alpha}-x^{(1)}_{-\alpha}=\mu_{\alpha}' , \ \ \alpha>0
%\]

%Since by assumption the element $h$ is regular ($h_\alpha\neq 1$ for any root $\alpha$)
%we obtain the formulae
%\[
%x^{(n)}_\alpha=\frac{\mu_\alpha''-a_{-\alpha}\mu_\alpha'-a_{-\alpha}(\mu_\alpha-\mu_{-\alpha})}{a_\alpha-a_{-\alpha}}
%\]
%\[
%x^{(n)}_{-\alpha}=\frac{\mu_\alpha''-a_{\alpha}\mu_\alpha'-a_{\alpha}(\mu_\alpha-\mu_{-\alpha})}{a_\alpha-a_{-\alpha}}
%\]
%\[
%x^{(1)}_\alpha=\frac{\mu_\alpha''-a_{-\alpha}\mu_\alpha'-(a_{\alpha}\mu_\alpha-a_{-\alpha}\mu_{-\alpha})}{a_\alpha-a_{-\alpha}}
%\]
%\[
%x^{(1)}_{-\alpha}=\frac{\mu_\alpha''-a_{\alpha}\mu_\alpha'-(a_{\alpha}\mu_\alpha-a_{-\alpha}\mu_{-\alpha})}{a_\alpha-a_{-\alpha}}
%\]
By the second line of \eqref{x-mu} we then obtain
%Now let us express the rest of $x$-variables in terms of $\mu$ and $a$:
\[
x^{(k)}_\alpha=\frac{a_\alpha\mu_{[\alpha]}^\prime+\mu_{[\alpha]}''+(a_{\alpha}^{-1}\mu_\alpha-a_{\alpha}\mu_{-\alpha})
+(a_\alpha-a_\alpha^{-1})(\mu_\alpha^{(1)}+\cdots+\mu_\alpha^{(k)})}{a_\alpha-a_{\alpha}^{-1}}
%\frac{\mu_\alpha''-a_{-\alpha}\mu_\alpha'-(a_{\alpha}\mu_\alpha-a_{-\alpha}\mu_{-\alpha})+(\mu^{(1)}_\alpha+\dots +\mu^{(k-1)}_\alpha)(a_\alpha-a_{-\alpha})}{a_\alpha-a_{-\alpha}}
\]
for $k=0,1,\ldots,n$. 
%\[
%x^{(k)}_{-\alpha}=\frac{\mu_\alpha''-a_{\alpha}\mu_\alpha'-(a_{\alpha}\mu_\alpha-a_{-\alpha}\mu_{-\alpha})+(\mu^{(1)}_{-\alpha}+\dots +\mu^{(k-1)}_{-\alpha})(a_\alpha-a_{-\alpha})}{a_\alpha-a_{-\alpha}}
%\]
%Here $\alpha>0$. 
Substituting \eqref{muexp} it is now easy to see that this is exactly what we wanted to prove.
\end{proof}
%%%%%%%%%%%%%%%%%%%%%%%%%%%%%%%%%
The proposition and \eqref{x-var} give an isomorphism 
\begin{equation}\label{open-sl}
S(\cO)_{reg}\simeq \bigl(( \cO_\ell^K \times \cO_1  \times  \dots \times\cO_n \times \cO_r^K)/M\times T^*A_{reg}\bigr)/W,
%(\cO_\ell^K \times \cO_1  \times  \dots \times\cO_n \times \cO_r^K)\times T^*A_{reg})/N_K(A)
\end{equation}
mapping $N_K(A)(x^{(0)},\ldots,x^{(n)},a)$ to the $W$-orbit of 
$(M(\mu^\prime,\mu^{(1)},\ldots,\mu^{(n)},\mu^{\prime\prime}),x_0^{(n)},a)$, which preserves the natural symplectic structures.
Here the finite discrete group $M=Z_K(A)\subset K$ acts diagonally via the coadjoint action, and $W=N_K(A)/M$ acts diagonally.
%Thus, $\cS(\cO)_{reg}$ is a symplectic  with the natural product symplectic structure on  the right side with $(x^{(n)}_0, h)\in T^*A_{reg}$. 
The quantum version of this isomorphism is described in \cite{SR}.
%This follows from the upcoming paper \cite{RS2} where quantum version of this isomorphism is described. So, we will omit the direct proof.

%%%%%%%%%%%%%%%%%%%%%%%%%%%%%%%%%%%%%%%%%%%%
\subsection{Quadratic Hamiltonians of open spin Calogero-Moser chain 
on the regular part of the phase space} \label{qu-Ham}
%%%%%%%%%%%%%%%%%%%%%%%%%%%%%%%%%%%%%%%%%%%%

In this section we compute the restriction of the Hamiltonian corresponding to the quadratic Casimir function on $\g^*$,
%Quadratic Hamiltonians are given by Casimir functions on $\g^*$. 
%For the Casimir function on k-th factor in $\g^{\times n}$ we have:
\[
H_2^{(k)}(x,g)=\frac{1}{2}(x^{(k)},x^{(k)})=\frac{1}{2}(x^{(k)}_0,x^{(k)}_0)+\sum_{\alpha>0} x^{(k)}_\alpha x^{(k)}_{-\alpha}
\]
to the regular part of $\cS(\cO)$ (see \eqref{bcase}) for $k=0,\ldots,n$, where $(x,g)\in\g^{*\times n+1}\times G^{\times n+1}$.
Here $(\cdot,\cdot)$ is the Killing form and 
$x_\alpha^{(i)}, x^{(i)}_0$ are the components of $x^{(i)}$ which were computed in the previous section on the regular part of the phase space.

We first consider the differences, which we will call the {\it boundary Knizhnik-Zamolodchikov-Bernard (bKZB) Hamiltonians},
\[
D_k=H_2^{(k)}-H_2^{(k-1)}\qquad\quad (1\leq k\leq n).
\]
%%%%%%%%%%%%%%%%%%%%%%%%%%%%%%%%%%%%%%%%
\begin{theorem} For the bKZB Hamiltonians we have the following formula:
\begin{equation}\label{Dk}
D_k=(\mu^{(k)}_0,x^{(n)}_0)-\sum_{l=1}^{k-1}(r_{lk}+r_{lk}^{\theta_l})+(\sum_\alpha K_\alpha\mu^{(k)}_{-\alpha}-\kappa_k)
+\sum_{l=k+1}^n(r_{kl}-r_{kl}^{\theta_k}).
\end{equation}
Here $r_{kl}$ for $k\not=l$ is Felder's rescaled dynamical $r$-matrix 
\begin{equation}\label{F-r-b}
r_{kl}=-\frac{1}{2}(\mu^{(k)}_0, \mu^{(l)}_0)+\sum_{\alpha}\frac{ \mu^{(k)}_{-\alpha}\mu^{(l)}_{\alpha}}{a_\alpha^2-1},
\end{equation}
$\theta_k$ is the transpose of the Chevalley involution $\theta$ acting on $\mu^{(k)}$,
\[
\kappa_k=\frac{1}{2} (\mu^{(k)}_0,\mu^{(k)}_0)+\sum_\alpha\frac{({\mu^{(k)}_\alpha})^2}{1-a_{\alpha}^2}
\]
is the core quadratic classical dynamical $k$-matrix 
and $K_\alpha$ is given by (\ref{Ka}).
\end{theorem}
%%%%%%%%%%%%%%%%%%%%%%%%%%%%%%%%%%%%%%%
\begin{proof}
The first step of the proof is the same as in the proof of Theorem \ref{thmradcyclic}, resulting in the expression
\begin{equation}\label{dd}
D_k=(\mu_0^{(k)},x_0^{(k-1)}+\frac{1}{2}\mu_0^{(k)})+\sum_\alpha x_\alpha^{(k-1)}\mu_{-\alpha}^{(k)}+\sum_{\alpha>0}\mu_\alpha^{(k)}\mu_{-\alpha}^{(k)}.
\end{equation}
%\[
%D_k=\frac{1}{2}(x^{(k+1)}_0-x^{(k)}_0, x^{(k+1)}_0+x^{(k)}_0)+\sum_{\alpha>0} (x^{(k+1)}_\alpha x^{(k+1)}_{-\alpha}-
%x^{(k)}_\alpha x^{(k)}_{-\alpha})
%\]
%The second group of identities between $x$ and $\mu$ implies
%\begin{equation}\label{Dk}
%D_k=\frac{1}{2}(\mu^{(k)}_0, x^{(k+1)}_0+x^{(k)}_0)+\sum_{\alpha}x^{(k)}_\alpha \mu^{(k)}_{-\alpha}+\sum_{\alpha>0}
%\mu^{(k)}_\alpha \mu^{(k)}_{-\alpha}
%\end{equation}
Now let us the formula (\ref{x-k-b}) for $x^{(k-1)}_\alpha$,
\begin{equation}\label{decompx}
\begin{split}
\sum_{\alpha} x^{(k-1)}_{\alpha}\mu^{(k)}_{-\alpha}&=\sum_{l=1}^{k-1}\sum_{\alpha}\frac{a_{\alpha}\mu^{(l)}_{\alpha}\mu_{-\alpha}^{(k)}-a_{\alpha}\mu^{(l)}_{-\alpha}\mu_{-\alpha}^{(k)}}{a_\alpha-a_{\alpha}^{-1}}\\
&+\sum_{\alpha} K_{\alpha}\mu^{(k)}_{-\alpha}+\sum_\alpha\frac{a_{\alpha}^{-1}\mu^{(k)}_{\alpha}\mu_{-\alpha}^{(k)}-a_{\alpha}\mu^{(k)}_{-\alpha}\mu_{-\alpha}^{(k)}}
{a_\alpha-a_{\alpha}^{-1}}\\
&+\sum_{l=k+1}^n\sum_\alpha\frac{a_{\alpha}^{-1}\mu_{-\alpha}^{(k)}\mu^{(l)}_{\alpha}-a_{\alpha}\mu_{-\alpha}^{(k)}\mu^{(l)}_{-\alpha}}{a_\alpha-a_{\alpha}^{-1}}.
\end{split}
\end{equation}
We express the different terms in the right hand side of \eqref{decompx} in terms of the dynamical $r$-matrix and $k$-matrix. 

Note first that
\[
r_{kl}^{\theta_k}=\frac{1}{2}(\mu_0^{(k)},\mu_0^{(l)})-\sum_\alpha\frac{\mu_\alpha^{(k)}\mu_\alpha^{(l)}}{a_\alpha^2-1}=r_{lk}^{\theta_l}.
\]
Then the terms in the right hand side of \eqref{decompx} with $l$ strictly smaller than $k$ can be rewritten as
\[
\sum_{\alpha}\frac{a_{\alpha}\mu^{(l)}_{\alpha}\mu_{-\alpha}^{(k)}-a_{\alpha}\mu^{(l)}_{-\alpha}\mu_{-\alpha}^{(k)}}{a_\alpha-a_{\alpha}^{-1}}=
-(r_{lk}+r_{lk}^{\theta_l})
\]
while the terms in the right hand side of \eqref{decompx} with $l$ strictly larger than $k$ reduce to
\[
\sum_\alpha\frac{a_{\alpha}^{-1}\mu_{-\alpha}^{(k)}\mu^{(l)}_{\alpha}-a_{\alpha}\mu_{-\alpha}^{(k)}\mu^{(l)}_{-\alpha}}{a_\alpha-a_{\alpha}^{-1}}=
(\mu_0^{(k)},\mu_0^{(l)})+(r_{kl}-r_{kl}^{\theta_k}).
\]
%Now we can use the formula for the dynamical $r$-matrix:
%\[
%\sum_{\alpha}\frac{a_{-\alpha}\mu^{(l)}_{\alpha}\mu^{(k)}_{-\alpha}}{a_\alpha-a_{-\alpha}}=r_{lk}+\frac{1}{2}(\mu^{(l)}_0, \mu^{(k)}_0)
%\]
%\[
%\sum_{\alpha}\frac{h_{-\alpha}\mu^{(l)}_{-\alpha}\mu^{(k)}_{-\alpha}}{a_\alpha-a_{-\alpha}}=-r_{lk}^{\theta_l}+\frac{1}{2}(\mu^{(l)}_0, \mu^{(k)}_0)
%\]
%\[
%\sum_{\alpha}\frac{a_{\alpha}\mu^{(l)}_{\alpha}\mu^{(k)}_{-\alpha}}{a_\alpha-a_{-\alpha}}=-r_{kl}-\frac{1}{2}(\mu^{(l)}_0, \mu^{(k)}_0)
%\]
Finally, for the middle term in \eqref{decompx} a direct computation shows that
\[
\sum_\alpha\frac{a_{\alpha}^{-1}\mu^{(k)}_{\alpha}\mu_{-\alpha}^{(k)}-a_{\alpha}\mu^{(k)}_{-\alpha}\mu_{-\alpha}^{(k)}}
{a_\alpha-a_{\alpha}^{-1}}=\frac{1}{2}(\mu_0^{(k)},\mu_0^{(k)})-\sum_{\alpha>0}\mu_\alpha^{(k)}\mu_{-\alpha}^{(k)}-\kappa_k.
\]
Substitute these formulas in \eqref{decompx}, then the resulting formula \eqref{dd} for $D_k$ becomes
%, take into account the identity
%\[
%\sum_\alpha\frac{a_{\alpha}\mu^{(k)}_{\alpha}-a_{-\alpha}\mu^{(k)}_{-\alpha}}{a_\alpha-a_{-\alpha}}\mu^{(k)}_{-\alpha}=\sum_{\alpha>0} \mu^{(k)}_{\alpha}\mu^{(k)}_{-\alpha}+\sum_\alpha\frac{{\mu^{(k)}_\alpha}^2}{1-a_{-\alpha}^2}
%\]
%and we arrive to the formula
\begin{equation*}
\begin{split}
D_k&=(\mu_0^{(k)},x_0^{(k-1)}+\mu_0^{(k)}+\mu_0^{(k+1)}+\cdots+\mu_0^{(n)})\\
&-\sum_{l=1}^{k-1}(r_{lk}+r_{lk}^{\theta_l})+(\sum_\alpha K_\alpha\mu^{(k)}_{-\alpha}-\kappa_k)
+\sum_{l=k+1}^n(r_{kl}-r_{kl}^{\theta_k}).
\end{split}
\end{equation*}
%\begin{align}
%D_k=\frac{1}{2}(\mu_0^{(k)}, x^{(k+1)}_0+x^{(k)}_0)+\sum_\alpha K_\alpha \mu^{(k)}_{-\alpha}
%-\sum_\alpha\frac{{\mu^{(k)}_\alpha}^2}{1-a_{-\alpha}^2}\\ -\sum_{l<k}(r_{lk}+r_{lk}^{\theta_l}) +\sum_{l>k}(\mu_0^{(k)}, \mu_0^{(l)})-\sum_{l>k}(-r_{kl}+r_{lk}^{\theta_l})
%\end{align}
%Taking into account the identities between $x_0$ and $\mu_0$ we have 
%\[
%\frac{1}{2}(\mu^{(k)}_0, x^{(k+1)}_0+x^{(k)}_0)+\sum_{l>k}(\mu_0^{(k)}, \mu_0^{(l)})=(\mu^{(k)}_0, x^{(n)}_0)-\frac{1}{2} (\mu^{(k)}_0,\mu^{(k)}_0)
%\]
By \eqref{x-var} this reduces to the formula (\ref{Dk}).
\end{proof}
%%%%%%%%%%%%%%%%%%%%%%%%%%%%%%%%%%%%%%%

The quantum versions of the boundary KZB Hamiltonians in the present context were obtained in \cite[\S 6]{SR}. It was extended to the case of non-split real semisimple Lie groups $G$ in \cite{RS}.
% In these papers it is also shown that the resulting dynamical $r$-matrices and $k$-

For the Hamiltonian $H^{(n)}_2$ we obtain by \eqref{extremes} the expression
\[
H^{(n)}_2=\frac{1}{2}(p,p)+\sum_{\alpha>0}\frac{(a_\alpha\mu_{[\alpha]}^\prime+\mu_{[\alpha]}^{\prime\prime}+a_\alpha(\mu_\alpha-\mu_{-\alpha}))(a_\alpha^{-1}\mu_{[\alpha]}^\prime+\mu_{[\alpha]}^{\prime\prime}+a_\alpha^{-1}(\mu_\alpha-\mu_{-\alpha}))}{(a_\alpha-a_{-\alpha})^2}
\]
%\[
%H^{(n)}_2=\frac{1}{2}(p,p)+\sum_{\alpha>0}\frac{(\mu_\alpha''-a_{-\alpha}\mu_\alpha'-a_{-\alpha}(\mu_\alpha-\mu_{-\alpha})(\mu_\alpha''-a_{\alpha}\mu_\alpha'-a_{\alpha}(\mu_\alpha-\mu_{-\alpha})}{(a_\alpha-a_{-\alpha})^2}
%\]
on $\cS(\cO)_{reg}$, where $\mu=\mu^{(1)}+\cdots+\mu^{(n)}$. Here we use notation $p=x^{(n)}_0$ for the cotangent vectors to $A_{reg}$ in formula (\ref{open-sl}).
Note that the potential term only depends on the restrictions $\pi(\mu^{(i)})$ of $\mu^{(i)}\in\g^*$ to $\mathfrak{k}$, since $\mu_\alpha-\mu_{-\alpha}=-(\pi(\mu))_{[\alpha]}$.
The radial component of the quantum quadratic Hamiltonian in the current open context was obtained in \cite[\S 6]{SR}.  
%and $a_\alpha$ is the function on $A_{reg}$ corresponding to the root $\alpha$.

%%%%%%%%%%%%%%%%%%%%%%%%%%%%%%%%%%%%%
\subsection{The superintegrability of the open spin CM chain}\label{op-sint}
%%%%%%%%%%%%%%%%%%%%%%%%%%%%%%%%%%%%%
In this section we will prove that Poisson commutative subalgebra of Hamiltonians
constructed in section \ref{oCM-Ham} defines a superintegrable system. Fix $\cO=(\cO_\ell^K,\cO_1,\ldots,\cO_n,\cO_r^K)\in
(\mathfrak{k}^*/K\times (\g^*/G)^{\times n}\times\mathfrak{k}^*/K$. We willl construct Poisson projections 
\begin{equation}\label{Sint-open-pr}
S(\cO)\stackrel{p_{1,\cO}}{\longrightarrow} \cP(\cO)\stackrel{p_{2,\cO}}{\longrightarrow} \cB(\cO)
\end{equation}
such that $p_\cO=p_{2,\cO}\circ p_{1,\cO}$ (see \eqref{p}), satisfying the desired properties.

Let $(\mathfrak{k}^*\times\g^{*\times n})\times_{(\g^*/G)^{\times n+1}}(\g^{*\times n}\times\mathfrak{k}^*)$ be the subset of 
$(\mathfrak{k}^*\times\g^{*\times n})\times(\g^{*\times n}\times\mathfrak{k}^*)$ consisting of elements
$(z_\ell,x_1,\ldots,x_n,y_1,\ldots,y_n,z_r)$ satisfying
\[
z_\ell\in -\pi(Gy_1),\quad x_i\in -Gy_{i+1}\,\, (1\leq i<n),\quad z_r\in -\pi(Gx_n).
\]
The gauge group $G_{n,K}$ acts on $(\mathfrak{k}^*\times\g^{*\times n})\times_{(\g^*/G)^{\times n+1}}(\g^{*\times n}\times\mathfrak{k}^*)$ by
\begin{equation}\label{g-act-P}
\begin{split}
(k_\ell, h_1, \dots, h_n, k_r)&(z_\ell,x_1,\dots, x_n,y_1,\dots, y_n, z_r)=\\
&= (Ad^*_{k_\ell}z_\ell, Ad^*_{h_1}x_1,\dots, Ad^*_{h_n}x_n,Ad^*_{h_1}y_1,\dots, Ad^*_{h_n}y_n, Ad^*_{k_r}z_r).
\end{split}
\end{equation}
%on $(\mathfrak{k}^*\times\g^{*\times n})\times (\g^{*\times n}\times\mathfrak{k}^*)$. 
Consider the resulting Poisson space
\[
\cP=\bigl((\mathfrak{k}^*\times\g^{*\times n})\times_{(\g^*/G)^{\times n+1}}(\g^{*\times n}\times\mathfrak{k}^*)\bigr)/G_{n,K}
\]
%with 
%\[
%(\mathfrak{k}^*\times\g^{*\times n})\times_{(\g^*/G)^{\times n+1}}(\g^{*\times n}\times\mathfrak{k}^*):=
%\{(z_\ell,x_1,\ldots,x_n,y_1,\ldots,y_n,z_r)\,\, | \,\, 
%\]
%with the $G_{n,K}$-action given by   
and define the Poisson map 
\[
p_1: T^*(G^{\times n+1})/G_{n,K}\rightarrow \cP
\] 
by
\begin{equation*}
\begin{split}
p_1\bigl(G_{n,K}(x,g)\bigr)&=G_{n,K}\bigl(\mu_L(x,g),\mu_R(x,g)\bigr)\\
&=G_{n,K}\bigl(\pi(x_0),x_1,\ldots,x_n,-Ad_{g_0^{-1}}^*(x_0),\ldots,-Ad_{g_{n-1}^{-1}}^*(x_{n-1}),-\pi(Ad_{g_n^{-1}}^*(x_n))\bigr).
\end{split}
\end{equation*}
Here $(x,g)=(x_0,\ldots,x_n,g_0,\ldots,g_n)\in\g^{*\times n+1}\times G^{\times n+1}\simeq T^*(G^{\times n+1})$.
%Because the group $G_{n, K}$ acts on $T^*(G^{\times n+1})$ by left and by right translations, and because this action is Hamiltonian, we have left and right moment maps $\mu_L, \mu_R$. Their combination gives the Poisson projection
%\[
%T^*(G^{\times n+1})\to \cP \subset (\kf^*\times {\g^*}^{\times n})\times_{(\g^*/G)^{\times n-1}} ({\g^*}^{\times n}\times \kf^*)
%\]
%where 
%\begin{eqnarray}
%\cP=\{(a,x_1,\dots, x_n;y_1,\dots, y_n, b)\in \kf^*\times {\g^*}^{\times n})\times ({\g^*}^{\times n}\times \kf^*| a\in \pi(Gy_1), \\ \nonumber Gx_1=-Gy_2, \dots Gx_{n-1}=-Gy_n, b\in -\pi(Gx_n)\}
%\end{eqnarray}
%This projection commutes with the gauge action of $G_{n,K}$ on $T^*(G^{\times n+1})$
%and the action
Define the Poisson projection 
\[
p_2: \cP\to (\g^*/G)^{\times n+1}
\]
by
\[
p_2\bigl(G_{n,K}(z_\ell,x_1,\dots, x_n,y_1,\dots, y_n, z_r)\bigr)=(-Gy_1,Gx_1,\dots, Gx_n),
\]
%\[
%G_{n,K}(z_\ell,x_1,\dots, x_n,y_1,\dots, y_n, z_r)\mapsto (-Gy_1, Gx_1, \dots, Gx_n)
%\]
with 
%the trivial action of $G_{n,K}$ and
the trivial Poisson structure on the target space. 
%Because all $G_{n,K}$ actions preserve Poisson structures we have two projections
%\begin{equation}\label{Pp}
% T^*(G^{\times n})/G_{n,K}\stackrel{\pi_1}{\rightarrow}\cP/G_{n,K}\stackrel{\pi_2}{\rightarrow} (\g^*/G)^{\times n}
% \end{equation}
%\subsubsection{The Poisson space $\cP(\cO)$}\label{p-space}

%Fix coadjoint orbits $\cO=(\cO_\ell^K,\cO_1,\ldots,\cO_n,\cO_r^K)\in\mathfrak{k}^*/K\times(\g^*/G)^{\times n}\times\mathfrak{k}^*/K$.
The restriction of the Poisson projection $p_1$ to the symplectic leaf $S(\cO)\subset T^*(G^{\times n+1})/G_{n,K}$ (see \eqref{bcase})
gives the Poisson projection
\[
p_{1,\cO}: S(\cO)\to \cP(\cO)
\]
where 
\[
\cP(\cO)=(\mu_L\times \mu_R)(\mu^{-1}(\cO))/G_{n,K},
\]
or, more explicitly, 
\begin{equation}\label{repP}
\begin{split}
\cP(\cO)=\bigl\{(z_\ell,x_1,\ldots,x_n, &y_1,\ldots,y_n,z_r)\in
(\mathfrak{k}^*\times\g^{*\times n})\times_{(\g^*/G)^{\times n+1}}(\g^{*\times n}\times \mathfrak{k}^*)\,\, | \,\, \\
&z_\ell\in\cO_\ell^K,\,x_1+y_1\in\cO_1,\ldots,x_n+y_n\in\cO_n,\, z_r\in\cO_r^K\bigr\}/G_{n,K}.
\end{split}
\end{equation}
%\begin{eqnarray}
%\cP(\cO)=\{(a,x_1,\dots, x_n;y_1,\dots, y_n, b)\in (\kf^*\times {\g^*}^{\times n})\times ({\g^*}^{\times n}\times \kf^*)| a\in \pi(Gy_1), \\ \nonumber Gx_1=-Gy_2, \dots Gx_{n-1}=-Gy_n, b\in -\pi(Gx_n); a\in \cO_\ell, x_1+y_1\in \cO_1, \dots, x_n+y_n\in \cO_n, b\in \cO_r\}/G_{n,K}
%\end{eqnarray}
%The action of $G_{n, K}$  is given by (\ref{g-act-P}). 
The generic fibers of this mapping are isotropic submanifolds in $\cS(\cO)$.

%\subsubsection{The space $\cB(\cO)$}\label{b-space}
%The space $\cB(\cO)$ was defined earlier as the spectrum of classical
%Hamiltonians. 

%It can be described more explicitly as follows
%\begin{align*}
%\cB(\cO)=\{( \cO^{(1)}, \dots, \cO^{(n)})\in (\g^*/G)^{\times n}|  p(\cO^{(1)})\subset \cO_1^K, \cO^{(2)}-\cO^{(1)}\subset \cO_1, \dots,\\  \cO^{(i+1)}-\cO^{(i)}\subset \cO_i, \dots, \cO^{(n)}-\cO^{(n-1)}\subset \cO_{n-1}, -p(\cO^{(n)})\subset \cO^K_2\}\subset (\g^*/G)^{\times n}
%\end{align*}

The restriction of the Poisson projection $p_2$ 
%in (\ref{Pp}) 
to $\cP(\cO)$ gives a surjective Poisson projection
\[
p_{2,\cO}: \cP(\cO)\to \cB(\cO),
\]
with $\cB(\cO)$ given by \eqref{cBO}. 
Clearly, the composition of $p_{2,\cO}\circ p_{1,\cO}: \cS(\cO)\rightarrow \cB(\cO)$ is the projection $p_\cO$ as given by \eqref{p}.  
%(\ref{pi-project}).

Now let us describe fibers of $p_{2,\cO}$ are symplectic leaves of $\cP(\cO)$. 

\begin{lemma} We have he following symplectomorphysm 
\begin{equation}\label{decomppOK}
p_{2,\cO}^{-1}((\cO^{(0)},\dots, \cO^{(n)}))\overset{\sim}{\longrightarrow}\cM(-\cO^{(0)},\cO_\ell^K)
\times\prod_{i=1}^n\mathcal{M}(-\cO^{(i)},\cO^{(i-1)},\cO_i)\times
\cM(\cO^{(n)},\cO_r^K)
\end{equation}
where symplectic spaces $\mathcal{M}(-\cO^{(i)},\cO^{(i-1)},\cO_i)$ are defined 
in (\ref{cM3}) and 
\begin{equation*}
\cM(\cO^\prime,\cO^K)=\{(x,z)\in\cO^\prime\times\cO^K\,\, | \,\, \pi(x)+z=0\}/K
\end{equation*}
Here $\cO\subset \g^*$ is a $G$-coadjoint orbit and $\cO^K\subset\mathfrak{k}^*$ is a $K$-coadjoint orbit.
It has a natural symplectic structure because it is the Hamiltonian reduction of $\cO\times \cO^K$ with
respect to the Hamiltonian diagonal action of $K$.  
\end{lemma}

\begin{proof}
Let $(\cO^{(0)},\ldots,\cO^{(n)})\in\cB(\cO)$. 
By a direct computation, the fiber $p_{2,\cO}^{-1}((\cO^{(0)},\dots, \cO^{(n)}))$ consists of the 
$G_{n,K}$-orbits in $\cP$ with representatives 
\begin{equation*}
(z_\ell,x_1,\ldots,x_n,y_1,\ldots,y_n,z_r)\in 
(\mathfrak{k}^*\times\mathfrak{g}^{*\times n})\times (\mathfrak{g}^{*\times n}\times\mathfrak{k}^*)
\end{equation*}
satisfying the following conditions
\begin{equation*}
\begin{matrix}
&z_\ell\in\cO_\ell^K\cap\pi(\cO^{(0)}),\quad &x_i+y_i\in\cO_i\quad (1\leq i\leq n),\,\,&z_r\in\cO_r^K\cap -\pi(\cO^{(n)}),\\
&-y_1\in \cO^{(0)},\quad &x_i\in\cO^{(i)}, -y_{i+1} \cO^{(i)}\in \, (1\leq i\leq n-1),\quad &x_n\in\cO^{(n)}.
\end{matrix}
\end{equation*}

Using this explicit description of the fiber, we can write it as a direct product of symplectic spaces. The isomorphism \eqref{cMiso} for symplectic spaces $\cM(\cO^{(1)},\cO^{(2)},\cO^{(3)})$ defined by \eqref{cM3} gives factors $\mathcal{M}(-\cO^{(i)},\cO^{(i-1)},\cO_i)$. The space
\begin{equation*}
\cM(\cO^\prime,\cO^K)=\{(x,z)\in\cO^\prime\times\cO^K\,\, | \,\, \pi(x)+z=0\}/K
\end{equation*}
with $K$ acting by the diagonal coadjoint action is symplectic because, see above. 
The isomorphism 
\[
\cM(\cO^\prime,\cO^K)\overset{\sim}{\longrightarrow} (\cO^\prime\cap\pi^{-1}(\cO^K))/K,
\]
completes the proof. The isomorphism maps
$G_{n,K}(z_\ell,x_1,\ldots,x_n,y_1,\ldots,y_n,z_r)\in p_{2,\cO}^{-1}((\cO^{(0)},\dots, \cO^{(n)}))$ to
\[
(K(\widetilde{y}_1,z_\ell),G(-x_1,-y_1,x_1+y_1),\ldots,G(-x_n,-y_n,x_n+y_n),K(\widetilde{x}_n,z_r)),
\]
with $\widetilde{y}_1\in Gy_1=-\cO^{(0)}$ such that $-\pi(\widetilde{y}_1)=z_\ell$ and $\widetilde{x}_n\in Gx_n=\cO^{(n)}$ such that $-\pi(\widetilde{x}_n)=z_r$.
\end{proof}

Note that for generic $\cO^\prime$, the symplectic space $\cM(\cO^\prime,\cO^K)$ is of dimension $\dim(\cO^K)$.

\begin{remark} In the compact case, the algebra of function on the fiber of $p_{2,\cO}$ has the algebra of endomorphisms of the vector space
\[
\textup{Hom}_K(V_{\lambda_0},U_{\nu_\ell})\otimes\textup{Hom}_G(V_{\lambda_1},V_{\lambda_0}\otimes V_{\mu_1})\otimes\cdots\otimes
\textup{Hom}_G(V_{\lambda_n},V_{\lambda_{n-1}}\otimes V_{\mu_n})\otimes\textup{Hom}_K(U_{\nu_r},V_{\lambda_n})
\]
as natural quantization, where the finite dimensional $G$-representation $V_{\lambda_i}$ \textup{(}resp. $V_{\mu_i}$\textup{)} corresponds to $\cO^{(i)}$ \textup{(}resp. $\cO_i$\textup{)} and the $K$-representation
$U_{\nu_\ell}$ \textup{(}resp. $U_{\nu_r}$\textup{)} corresponds to $\cO_\ell^K$ \textup{(}resp. $\cO_r^K$\textup{)}, compare with Remark \ref{cyclicquantum} in the cyclic case. 
For details see \cite{SR} \textup{(}which treats the noncompact case\textup{)} and \cite{RSQCD2}.
\end{remark}

\begin{lemma} Dimensions of spaces $\cB(\cO)$ and $\cP(\cO)$ are 
\[
\dim(\cB(\cO))=(n+1)r, \ \  \dim(\cP(\cO))=\dim(\cO)-2nr
\]
where we define $\dim(\cO)$ as $\dim(\cO_\ell^K)+\sum_{i=1}^n\dim(\cO_i)+\dim(\cO_r^K)$.
\end{lemma}

\begin{proof}
The proof of the dimension formula for $\cB(\cO)$ is completely similar to the periodic case. 
It is enough to consider large orbits. For the dimension of $\cP(\cO)$ we have:
\[
\dim(\cP(\cO))=\dim(\cB(\cO))+\dim\bigl(p_{2,\cO}^{-1}((\cO^{(0)},\ldots,\cO^{(n)}))
\]
and by \eqref{decomppOK} the dimension of
$p_{2,\cO}^{-1}((\cO^{(0)},\ldots,\cO^{(n)}))$ is equal to
\begin{equation*}
\begin{split}
\dim(\cM(-\cO^{(0)},\cO_\ell^K))&+\sum_{i=1}^n\dim(\cM(-\cO^{(i)},\cO^{(i-1)},\cO_i))+\dim(\cM(\cO^{(n)},\cO_r^K))=\\
&=\dim(\cO_\ell^K)+\sum_{i=1}^n(\dim(\cO_i)-2r)+\dim(\cO_r^K)=\dim(\cO)-2nr.
\end{split}
\end{equation*}
This finishes the proof.
\end{proof}

We now have the following main result of this section.
\begin{theorem}
The Hamiltonian system generated by any Hamiltonian for the open spin CM chain described in section \ref{oCM-Ham} is superintegrable with the superintegrable structure described by the surjective Poisson maps
\[
\cS(\cO) \stackrel{p_{1,\cO}}{\longrightarrow} \cP(\cO) \stackrel{p_{2,\cO}}{\longrightarrow} \cB(\cO)
\]
as introduced earlier in this section.
\end{theorem}
Recall that $\cO_i\neq \{0\}$ for all $i=0,1, \dots, n$. 

\begin{proof}
We already verified most of the conditions. What remains to show is the matching of dimensions,
\begin{equation}\label{lastcheck}
\dim(\cS(\cO))=\dim(\cP(\cO))+\dim(\cB(\cO)).
\end{equation}
For the collection $\cO=(\cO_\ell^K, \cO_1,\dots, \cO_{n}, \cO_r^K)$ of coadjoint orbits we write $\dim(\cO)$ for 
the sum  of the dimensions of the coadjoint orbits. 
 
For the dimension
of the symplectic leaf $\cS(\cO)$ we have, using \eqref{open-sl},
\[
\dim(\cS(\cO))=2r+\dim(\cO),
\]
with $r$ the rank of $\g$. For $\cP(\cO)$ we obtain for generic $(\cO^{(0)},\ldots,\cO^{(n)})\in\cB(\cO)$,

Because
\[
\dim(\cB(\cO))=(n+1)r,
\]
we have
\begin{equation*}
\begin{split}
\dim(\cP(\cO))+\dim(\cB(\cO))&=\dim(\cO)-2nr+2\dim(\cB(\cO))\\
&=\dim(\cO)+2r\\
&=\dim(\cS(\cO)),
\end{split}
\end{equation*}
as desired.

\end{proof}

\subsection{Constructing solutions by the projection method and angle variables } \label{op-proj-dyn}
%%%%%%%%%%%%%%%%%%%%%%%%%%%%%%%
Let $\mathcal{H}$ be a $G$-invariant function on $\g^*$ and $\mathcal{H}^{(i)}$ for $i=0,\ldots,n$ the associated $G_n$-invariant function $(x,g)\mapsto \mathcal{H}(x_i)$ on $T^*(G^{\times n+1})$ (cf. section \ref{solutions}). 
%Let $H\in C(\g^*)^G$ be a $G$-invariant function on $\g^*$ and $H(x_i)$ be a function on $\g^{\times n}\times G^{\times n}$ which is $H$ on the $i$-th factor $\g^*$ and constant on others. 
The Hamiltonian flow generated by $\mathcal{H}^{(i)}$ on $T^*(G^{\times n+1})\simeq \g^{*\times n+1}\times G^{\times n+1}$ was already described in section \ref{solutions}. The flow line 
passing through $(x,g)$ at $t=0$ is
\begin{equation}\label{flow-open}
(x(t_i), g(t_i))=(x_0, \dots, x_n, g_0, \dots,g_{i-1}, e^{\nabla\mathcal{H}(x_i) t_i}g_i,g_{i+1} \dots, g_n).
%(x(t_i),g(t_i))=(x_1,\dots, x_n, g_1, \dots, e^{t_i\nabla H_k(x_i)}g_i, \dots, g_n)
\end{equation}
The corresponding Hamiltonian flow on the symplectic leaf $\cS(\cO)\subset T^*(G^{\times n+1})/G_{n,K}$ 
%(see \eqref{bcase}) 
is obtained by projecting the flow (\ref{flow-open}) 
to $T^*(G^{\times n+1})/G_{n,K}$ and restricting it 
%with the subsequent restriction 
to $\cS(\cO)$. Thus, we reformulated the problem of solving nonlinear differential equations of motion for
open spin Calogero-Moser chains to a problem of linear algebra. This is a version of the original projection method which goes back to earlier papers on Calogero-Moser type systems \cite{OP}. 

Now let us describe angle variables for this integrable dynamics. We use the notations from section \ref{solutions}. Fix $(x,g)\in \g^{\prime*\times n+1}\times G^{\times n+1}$. 
%Fix a Borel subalgebra $\bf\in\g$ and let $\hf\in \g$ be the corresponding Cartan subalgebra. 
For $i=1,\dots, n$ define elements $s_i\in G$ by the condition $Ad^*_{s_i}(x_i)\in \mathfrak{a}^*_{+}$. 
As in the periodic case (see section \ref{solutions}) it is defines up to $s_i\mapsto a_is_i$ with $a_i\in H\subset G$, where $H\subset G$ is the Cartan subgroup containing $A$. 
%Gauge transofrmations act on $s_i$ as $s_i\mapsto s_ih^{-1}_i$. 
%Thus 
%$s_ig_is_{i+1}^{-1}\in G$ is gauge invariant for $i=1,\dots, n-1$. 
Define $s_0\in K$ such that $Ad_{s_0}^*(x_0)\vert_{\mathfrak{p}}\in\mathfrak{a}_{+}^*$ (where we view $\mathfrak{a}_+^*$ now as subset of $\mathfrak{p}^*$ in the natural manner).
The element $s_0$ is defined up to $s_0\mapsto ms_0$ with $m\in M=Z_K(A)$. Similarly, we define $s_{n+1}\in K$ such that $Ad_{s_{n+1}}^*(x_n)\vert_{\mathfrak{p}}\in
\mathfrak{a}_+^*$. 

Choose finite dimensional representations $V_0,V_1,\dots, V_n$ of $G_{\mathbb{C}}$, $H_{\mathbb{C}}$-weight vectors $v_i\in V_i$ of weight $\lambda_{i+1}$ for $0\leq i<n$ and
$H_{\mathbb{C}}$-weight vectors $u_j^*\in V_j^*$ of weight $-\lambda_j$ for $0\leq j\leq n$. Finally, we choose $M$-invariant vectors $u_0^*\in V_0^*$ and $v_n\in V_n$ (i.e.,
$mu_0^*=u_0^*$ and $mv_n=v_n$ for all $m\in M$).
%$v_1\in V_1,\dots, v_{n-1}\in V_n$ and $u^*_0\in V_1^*, u^*_1\in V_2^* \dots, u^*_{n-1}\in V^*_n$ of weights $\lambda_1, \dots, \lambda_{n-1}$ and $-\lambda_0,-\lambda_1, \dots, -\lambda_{n-1}$ respectively. 
%Let 
Define 
\begin{equation}\label{fopen}
f_{u,v}(x,g)=u^*_0(s_0g_0s^{-1}_1v_0)u^*_1(s_1g_1s^{-1}_2v_1)\cdots u_{n-1}^*(s_{n-1}g_{n-1}s_n^{-1}v_{n-1})u_n^*(s_ng_ns_{n+1}^{-1}v_n).
\end{equation}
It is an easy check that $f_{u,v}(x,g)$ is a well defined $G_{n,K}$-invariant function on $\g^{\prime*\times n+1}\times G^{\times n+1}$. 
%\begin{equation}\label{fopen}
%f_{u,v}(x,g)=u^*_0(s_0g_0s^{-1}_1v_1)u^*_1(s_1g_1s^{-1}_2v_2)\dots
%u^*_{\lambda_{n-1}}(s_ng_nu)
%\end{equation}
%where $u\in V_n$ is a spherical vector, $ku=u$ for any $k\in K$. 
%Clearly this expression is gauge invariant and is invariant with respect to $s_i\mapsto a_is_i$. Thus, it is a function on $T^*G^{\times n}/G^n$.

%For a finite dimensional representation $V$ and weight vectors $u_\mu^*\in V^*$ and $v_\lambda\in V$ 
%we have the identity
%\[
%u^*_\mu(s_ie^{t\nabla H(x_i)}g_is_{i+1}^{-1}v_\lambda)=e^{t<\mu, \nabla H(y_i)>}u^*_\mu(s_ig_is_{i+1}^{-1}v_\lambda)
%\]
Similarly as in the periodic case (see section \ref{solutions}) we then have for $i=1,\ldots,n$,
\begin{equation}\label{trivialflow}
f_{u,v}(x(t_i),g(t_i))=e^{t_i\lambda_i(\nabla\mathcal{H}(y_i))}f_{u,v}(x,g)
\end{equation}
with $y_i=Ad_{s_i}^*(x_i)\in\mathfrak{a}_+^*$. Logarithms of these functions thus evolve linearly, and hence give rise to angle variables for the Hamiltonians $\mathcal{H}^{(i)}$ on $\cS(\cO)\cap (\g^{\prime*\times n+1}\times G^{\times n+1})/G_{n,K}$.

For $i=0$ we need to restrict further to $(x,g)\in \g^{\prime*\times n+1}\times G^{\times n+1}$ with $x_0\in\mathfrak{p}$, and assume that $u_0^*\in V_0^*$ is not only $M$-invariant but also a $H_{\mathbb{C}}$-weight vector, say of weight $-\lambda_0$. In this case $\Ad_{s_0}^*(x_0)=y_0\in\mathfrak{a}_+^*$ and hence
%Also for the last factor in (\ref{fopen}) we have 
%\[  
%u^*_{\lambda_{n-1}}(s_ne^{t\nabla H(x_n)}g_nu)=e^{t<\mu, \nabla H(y_n)>}u^*_{\lambda_{n-1}}(s_ng_nu)
%\]
\[  
u^*_0(s_0e^{t_0\nabla\mathcal{H}(x_0)}g_0s_1^{-1}v_0)=e^{t_0\lambda_0(\mathcal{H}(y_0))}u^*_{0}(s_0g_0s_1^{-1}v_0).
\]
As a consequence \eqref{trivialflow} then also holds true for $i=0$, and the logarithm of $f_{u,v}(x,g)$ becomes a linear functions of time $t_0$. 

%In particular, these functions give rise to angle variables for the Hamiltonian $\mathcal{H}^{(0)}$ 
%on $\cS(\cO)\cap (\g^{\prime*\times n+1}\times G^{\times n+1})/G_{n,K}$ in the special case that $\cO_\ell^K=\{0\}$.
%Thus, maximal set of linear independent functions in the span of logarithms of $f_{u,v}(x,g)$ is a collection of action-angle variables for
%our system. 

\section{A Liouville integrable example of a periodic spin Calogero-Moser example for orbits of rank $1$. }\label{concl}

\subsection{} Let us briefly discuss a particular case of periodic spin CM chain corresponding to
$G=SL_N(\mathbb{R})$ with rank one orbits $\cO_k$. This case is related to the original 
paper \cite{GH} where spin CM systems were first introduced. 

Take $\mathfrak{a}\subset\mathfrak{sl}_N$ the Cartan subalgebra consisting of diagonal matrices, and denote the roots by 
$\{\epsilon_i-\epsilon_j\}_{i\not=j}\subset\mathfrak{a}^*$ with $\epsilon_i\in\mathfrak{a}^*$ the linear functional picking out the $i^{th}$ diagonal entry.
We identify $\mathfrak{sl}_N$ with its dual via the Killing form $(x,y)=2N\,\textup{Tr}(xy)$. Then for $p\in\mathfrak{a}^*\simeq\mathfrak{a}$
we have $(p,p)=2N\sum_{i=1}^Np_i^2$, with $p_i$ the $i^{th}$ diagonal entry of the diagonal matrix $p$. For $y\in\mathfrak{sl}_N^*\simeq\mathfrak{sl}_N$ 
and $i\not=j$ we have $y_{\epsilon_i-\epsilon_j}=\sqrt{2N}y_{ij}$, with $y_{ij}$ the $(i,j)^{th}$ entry of the matrix $y$.

For $\xi\in\mathbb{R}$ set
\[
\cO^{(\xi)}=\Bigl\{x-\frac{\xi}{N}\textup{id}_N\,\,\, | \,\,\, x \textup{ is a rank one $N\times N$ matrix with } \,\, \textup{Tr}(x)=\xi\,\Bigr\}.
\]
Then $\cO^{(\xi)}$ is a coadjoint orbit in $\mathfrak{sl}_N\simeq\mathfrak{sl}_N^*$ of dimension $2(N-1)$. 

Viewing elements in $\mathbb{R}^N$ as column vectors, we have a natural mapping
\begin{equation}\label{param}
\bigl\{(a,b)\in\mathbb{R}^N\times\mathbb{R}^N \,\, | \,\, a^tb=\xi\bigr\}/\mathbb{R}^\times\overset{\sim}{\longrightarrow}\cO^{(\xi)},\qquad
\mathbb{R}^\times(a,b)\mapsto ba^t-\frac{\xi}{N}\textup{id}_N
\end{equation}
where $\lambda\in\mathbb{R}^\times$ acts by $(a,b)\mapsto (\lambda a, \lambda^{-1}b)$ and $a^t$ is the transpose of $a\in\mathbb{R}^N$. Because of the rank one condition, this is an isomorphism. It is
easy to check that this is symplomorphism, with the Poisson brackets of the coordinate functions $a_i$ and $b_j$ of $(a,b)\in\mathbb{R}^N\times\mathbb{R}^N$ given by
\[
\{b_i,a_{j}\}=\delta_{ij}, \ \ \{a_{i},a_{j}\}=0=\{b_i,b_j\}=0.
\]
The value of the quadratic Casimir function 
$y\mapsto \frac{(y,y)}{2N}=\sum_{i,j=1}^Ny_{ij}y_{ji}$
on $\mathcal{O}^{(\xi)}$ is easily computed using \eqref{param}:
\[
\sum_{i,j=1}^N\mu_{ij}\mu_{ji}=\xi^2\Big(1-\frac{1}{N}\Bigr),\qquad\qquad \mu\in\cO^{(\xi)}.
\]
Here we use notations $\mu=y|_{\cO^{(\xi)}}$.

\subsection{}The quadratic $n$-th Hamiltonian $H_2^{(n)}$ in radial coordinates, rescaled by a factor $2N$, then is 
\begin{equation}\label{Cas-Ham}
H_2=\frac{1}{2}\sum_{i=1}^Np_i^2-\sum_{i<j}\frac{\mu_{ij}\mu_{ji}}{2\textup{sh}^2(q_i-q_j)}
\end{equation}
%\begin{equation}\label{Cas-Ham}
%H=\frac{1}{2}\sum_{i=1}^Np_i^2+\sum_{i<j}\frac{f_{ij}f_{ji}}{4sh^2(\frac{q_i-q_j}{2})}
%\end{equation}
(see \eqref{Hamper}),
where $q_i=\epsilon_i(\log(a))$ and 
\[
\mu_{ij}=\sum_{k=1}^n \mu_{ij}^{(k)}.
\]

We now consider the Hamiltonian \eqref{Cas-Ham} on $\cS(\cO)_{reg}\simeq (\nu_{\mathcal{O}}^{-1}(0)/H\times T^*A_{reg})/W$ (see \eqref{radred}) with $\cO=(\cO^{(\xi_1)},\ldots,\cO^{(\xi_n)})$
a collection of $n$ rank one orbits. Here $H\subset\textup{SL}_N(\mathbb{R})$ is the Cartan subgroup of diagonal matrices and 
$\nu_{\cO}^{-1}(0)\subset\cO^{(\xi_1)}\times\cdots\times\cO^{(\xi_n)}$ consists of the $n$-tuple of rank one matrices
\begin{equation}\label{abpar}
(\mu^{(1)},\ldots,\mu^{(n)})=\Bigl(b^{(1)}a^{(1)t}-\frac{\xi_1}{N}\textup{id}_N,\ldots,b^{(N)}a^{(N)t}-\frac{\xi_N}{N}\textup{id}_N\Bigr)
\end{equation}
where the diagonal action of $h\in H$ is given by $a_i^{(k)}\to h_ia_i^{(k)}$, $b_j^{(k)}\to h_j^{-1}b_j^{(k)}$, and vectors $a^{(k)},b^{(k)}\in\mathbb{R}^N$ satisfy the relations
\begin{equation}\label{g-Casnew}
\sum_{i=1}^Na_i^{(k)}b_i^{(k)}=\xi_k\quad (1\leq k\leq n),\qquad \sum_{k=1}^na_i^{(k)}b_i^{(k)}=\frac{\bm{\xi}}{N}\quad (1\leq i\leq N).
\end{equation}
Here $\bm{\xi}=\sum_{k=1}^n\xi_k$. 
In other words, $\mu^{(k)}_{ij}=b_i^{(k)}a_j^{(k)}-\delta_{ij}\xi_k/N$, where $b_i^{(k)}, a_j^{(k)}$ are as above.

In terms of the variables $a^{(k)}$ and $b^{(k)}$ the Hamiltonian $H_2$ on $\cS(\cO)_{reg}$ can be rewritten as
\begin{equation}\label{spin-CM-ab}
H_2=\frac{1}{2}\sum_{i=1}^Np_i^2-\sum_{i<j}\frac{\sum_{k,\ell=1}^nb_i^{(k)}a_j^{(k)}b_j^{(\ell)}a_i^{(\ell)}}{2\textup{sh}^2(q_i-q_j)}.
\end{equation}

\subsection{} Here we will rewrite the Hamiltonian (\ref{spin-CM-ab}) in terms of variables
attached to each $q_i$. It is natural to think of these variable as spin variables attached to a 
one dimensional particle with the position $q_i$. They are defined as follows.

For $\xi\in \mathbb{R}$ denote by $\widetilde{\cO}^{(\xi)}$ the rank one coadjoint $\textup{SL}_n(\mathbb{R})$-orbit defined as
\[
\widetilde{\cO}^{(\xi)}=\Bigl\{x-\frac{\xi}{n}\textup{id}_n\,\,\, | \,\,\, x \textup{ is a rank one $n\times n$ matrix with}\,\, \textup{Tr}(x)=\xi\,\Bigr\},
\]
and set 
\[
\widetilde{\cO}:=\underbrace{(\widetilde{\cO}^{(\bm{\xi}/N)},\ldots,\widetilde{\cO}^{(\bm{\xi}/N)})}_N.
\]
The coadjoint action of the Cartan subgroup $\widetilde{H}\subset \textup{SL}_n(\mathbb{R}$) is Hamiltonian and gives the moment map 
\[
\widetilde{\nu}_{\widetilde{\cO}}: \underbrace{\widetilde{\cO}^{(\bm{\xi}/N)}\times\cdots\times\widetilde{\cO}^{(\bm{\xi}/N)}}_N\rightarrow\widetilde{\mathfrak{a}}^*,\qquad
(g_1,\ldots,g_N)\mapsto\bigl(g_1+\cdots+g_N)_0
\]
where $\widetilde{\mathfrak{a}}=\textup{Lie}(\widetilde{H})$.
%and we by $(g_1,\ldots,g_N)$ elements of $\widetilde{\cO}^{(\bm{\xi}/N)}\times\cdots\times\widetilde{\cO}^{(\bm{\xi}/N)}$.

Finally, consider the traceless diagonal $n\times n$-matrix
\[
t_{\underline{\xi}}=\textup{diag}\Bigl(\xi_1-\frac{\bm{\xi}}{n},\ldots,\xi_n-\frac{\bm{\xi}}{n}\Bigr)\in\widetilde{\mathfrak{a}}.
\]
From the above we immediately have the following statement.
\begin{lemma} We the following isomorphism of $2(n-1)(N-1)$-dimensional symplectic varieties
\begin{equation*}
\begin{split}
\nu_{\cO}^{-1}(0)/H&\overset{\sim}{\longrightarrow}\widetilde{\nu}_{\widetilde{\cO}}^{-1}(t_{\underline{\xi}})/\widetilde{H},\\
H(\mu^{(1)},\ldots,\mu^{(n)})&\mapsto \widetilde{H}(g^{(1)},\ldots,g^{(N)}),
\end{split}
\end{equation*}
Here $\mu^{(k)}_{ij}=b_i^{(k)}a_j^{(k)}-\delta_{ij}\xi_k/N$ and the local spin variables $g^{(i)}_{k\ell}$ are %defined in terms of the $(a,b)$-coordinates (see \eqref{abpar}) by
\[
g_{k\ell}^{(i)}=b_i^{(k)}a_i^{(\ell)}-\delta_{k\ell}\frac{\bm{\xi}}{Nn}.
\]
\end{lemma}
It is easy to check that if $i\neq j$ the following identity holds:
\begin{equation}\label{g-spin}
\sum_{k,\ell=1}^n g^{(i)}_{k\ell}g_{\ell k}^{(j)}=\mu_{ij}\mu_{ji}-\frac{\bm{\xi}^2}{N^2n},
%(g_i)_\beta^\alpha(g_j)_\alpha^\beta=f_{ij}f_{ji}-\frac{c^2}{nN}
\end{equation}
%\begin{equation}\label{g-spin}
%\sum_{\alpha, \beta=1}^n (g_i)_\beta^\alpha(g_j)_\alpha^\beta=f_{ij}f_{ji}-\frac{c^2}{nN}
%\end{equation}
%while for $i=j$ we have
%\begin{equation}\label{g-Cas}
%\sum_{k,\ell=1}^ng^{(i)}_{k\ell}g^{(i)}_{\ell k}=\frac{{\bm{\xi}}^2}{N^2}\bigl(1-\frac{1}{n}\bigr)
%\end{equation}
%\begin{equation}\label{g-Cas}
%\sum_{\alpha, \beta=1}^n (g_i)_\beta^\alpha(g_i)_\alpha^\beta=\frac{c^2}{n^2}(1-\frac{1}{N})
%\end{equation}

%Note that we proved an isomorphism of symplectic varieties:
%\[
%(\cO_{N_1}\times \dots\times \cO_{N_n})//A\simeq (\cO^{(n)}\times \dots\times \cO^{(n)})//A^{(n)}
%\]
%where $\cO^{(n)}$ is the rank one coadjoint orbit in $sl_n^*$ with the Casimir given by (\ref{g-Cas}) and $A^{(n)}\subset SL_n$ is maximally non-compact Cartan subgroup.

Thus, we can rewrite the Hamiltonian (\ref{Cas-Ham}) in terms of spin variables from $\widetilde{\nu}_{\widetilde{\cO}}^{-1}(t_{\underline{\xi}})/\widetilde{H}\times T^*A_{reg}$ as
\begin{equation}\label{spin-spin-CM}
H_2=\frac{1}{2}\sum_{i=1}^Np_i^2+\sum_{i<j}\frac{\textup{Tr}(g^{(i)}g^{(j)})+\frac{\bm{\xi}^2}{N^2n}}{2\textup{sh}^2(q_i-q_j)}.
\end{equation}

This Hamiltonian describes $N$ classical particles each 
carrying a "spin" from a rank one coadjoint orbit in $\mathfrak{sl}_n^*$ with the Casimir value given by 
\begin{equation}
\sum_{\alpha, \beta=1}^n (g_i)_\beta^\alpha(g_i)_\alpha^\beta=\frac{c^2}{n^2}(1-\frac{1}{N})
\end{equation}

The system is Liouville integrable since we constructed $n(N-1)$ integrals
for the periodic spin chain earlier (see the proof of Theorem \ref{siperiodic}).
%{\color{red} Provided that the dimension counts are fine for rank one orbits.}

Integrable system described above are closely related \cite{GH} and \cite{KBBT}.

\subsection{} This project, together with results of \cite{AR}, is the first step towards
constructing superintegrable systems on moduli spaces of flat connections on a surface where on part of the boundary the gauge group $G$ is constrained to $K$. When the boundary gauge group is not constrained, corresponding integrable systems are described in \cite{AR}.  We expect that such moduli spaces have the structure of a cluster variety similar to the one described in \cite{FG}.
It would be interesting to 
to extend the construction of spin CM chains to the elliptic case as it was done for $N=1$ in \cite{KBBT}.

\appendix

\section{Comparison with the $n=2$ case from \cite{R3}}\label{R3}

Consider the periodic spin CM chain from section \ref{pCM} for $n=2$.
The symplectic leaves of $T^*(G^{\times 2})/G_2$ are then
\begin{equation}\label{s1}
\cS(\cO_1,\cO_2)=\{(x_1,x_2,g_1,g_2)\,\, |\,\, x_1-Ad^*_{g_2^{-1}}(x_2)\in \cO_1, \ \ x_2-Ad^*_{g_1^{-1}}(x_1)\in \cO_2\}/G_2
\end{equation}
%\begin{equation}\label{s1}
%\cS(\cO_1,\cO_2)=\{(x_1,x_2,g_1,g_2)|x_2-Ad^*_{g_1^{-1}}(x_1)\in \cO_2, \ \ x_1-Ad^*_{g_2^{-1}}(x_2)\in \cO_1\}/G\times G
%\end{equation}
where $\cO_1,\cO_2$ are coadjoint orbits in $\g^*$, relative to the gauge action 
\[
(h_1,h_2)(x_1,x_2,g_1,g_2)=(Ad_{h_1}^*(x_1), Ad_{h_2}^*(x_2), h_1g_1h_2^{-1},h_2g_2h_1^{-1}).
\]
%\[
%(h_1,h_2)(x_1,x_2,g_1,g_2)=(Ad_{h_2}^*(x_1), Ad_{h_1}^*(x_2), h_2g_1h_1^{-1},h_1g_2h_2^{-1}).
%\]

%The mapping $(x_1,x_2,g_1,g_2)\mapsto (-x_1,x_2,g_1^{-1},g_2)$ is a symplectomorphism of $T^*(G\times G)$. It brings the gauge action of $G\times G$ on $T^*(G\times G)$ to the lift of the action  of $G\times G$ on itself by the diagonal left translations of $G\times\{1\}\subset G\times G$
%and by diagonal right translations of $\{1\}\times G\subset G\times G$ to $T^*(G\times G)$.
In \cite[\S 3 \& App. C]{R3} the following Hamiltonian action of $G_2$ on $T^*(G^{\times 2})$ is considered,
\begin{equation}\label{action*}
(h_1,h_2)_*(x_1,x_2,g_1,g_2)=(Ad_{h_1}^*(x_1),Ad_{h_1}^*(x_2),h_1g_1h_2^{-1},h_1g_2h_2^{-1}),
\end{equation}
with corresponding moment map $\mu_*: T^*(G^{\times 2})\rightarrow\g^{*\times 2}$ given by
\[
\mu_*(x_1,x_2,g_1,g_2)=(x_1+x_2,-Ad_{g_1^{-1}}^*(x_1)-Ad_{g_2^{-1}}^*(x_2)).
\]
The corresponding symplectic leaves are 
\begin{equation*}
\begin{split}
\cS_*(\cO_1,\cO_2)&=\mu_*^{-1}(\cO_1,\cO_2)/G_2\\
&=\bigl\{(x_1,x_2,g_1,g_2) \,\, | \,\, x_1+x_2\in\cO_1,\,\,-Ad_{g_1^{-1}}^*(x_1)-Ad_{g_2^{-1}}(x_2)\in\cO_2\bigr\}/G_2,
\end{split}
\end{equation*}
with the gauge group $G_2$ now acting by \eqref{action*}. These symplectic leaves were used in \cite{R3}. They are related to the symplectic leaves $\cS(\cO_1,\cO_2)$ in the following way.

Consider the map $\psi: T^*(G^{\times 2})\rightarrow T^*(G^{\times 2})$, defined by 
\[
\psi(x_1,x_2,g_1,g_2)=(-x_1,Ad_{g_1}^*(x_2),g_1,g_1g_2g_1).
\]
Then $\psi$ is $G_2$-equivariant,
\[
\psi((h_1,h_2)(x_1,x_2,g_1,g_2))=(h_1,h_2)_*\psi(x_1,x_2,g_1,g_2),
\]
and the resulting map on the $G_2$-orbits restricts to an isomorphism 
\begin{equation*}
\cS(\cO_1,\cO_2)\overset{\sim}{\longrightarrow}\cS_*(\cO_2,\cO_1).
\end{equation*}
%$\cS(\cO_1,\cO_2)\overset{\sim}{\

%This mapping gives an isomorphism of Poisson manifolds $T^*(G\times G)/(G\times G)\simeq G\backslash T^*(G\times G)/G$. It also give the identification of (\ref{s1}) with the symplectic leaf
%\begin{equation}\label{s2}
%\cS(\cO_1,\cO_2)=G\backslash \{(x_1,x_2,g_1,g_2)|x_2-Ad^*_{g_1^{-1}}(x_1)\in \cO_2, \ \ x_1-Ad^*_{g_2^{-1}}(x_2)\in \cO_1\}/G
%\end{equation}
%of $G\backslash T^*(G\times G)/G$ which was used in \cite{R3}.

\end{document}